\documentclass[11pt]{article}

\usepackage{times}
\usepackage[paperwidth=199.8mm,
paperheight=297mm,centering,hmargin=25mm,vmargin=25mm]{geometry}
\usepackage{authblk} 
\usepackage[bottom]{footmisc} 

\usepackage[titletoc,title]{appendix} 
\usepackage{changepage}

\usepackage[activate={true,nocompatibility},final,tracking=true,kerning=true,spacing=true,factor=1100,stretch=10,shrink=10]{microtype}

\usepackage{hyperref}
\hypersetup{colorlinks=true,citecolor=myblue,linkcolor=myblue,
filecolor=myblue,urlcolor=myblue,breaklinks=mytrue}
\usepackage{url}

\usepackage[dvipsnames]{xcolor}
\usepackage{color}
\usepackage{framed}

\definecolor{shadecolor}{rgb}{0.9,0.9,0.9}
\definecolor{mylightgray}{RGB}{150,150,150}

\definecolor{myblue}{RGB}{0, 68, 116}
\definecolor{mycyan}{RGB}{0, 97, 91}
\definecolor{mygreen}{RGB}{2, 102, 1}
\definecolor{myorange}{RGB}{240, 102, 0}
\definecolor{myred}{RGB}{172, 23, 0}
\definecolor{mymagenta}{RGB}{140,16,73}

\usepackage{mathtools}
\usepackage{amsmath}
\usepackage{amssymb}
\usepackage[shortlabels]{enumitem}
\usepackage{graphicx,epic,eepic,epsfig,latexsym,verbatim}
\usepackage{dsfont}
\usepackage{mathrsfs}

\usepackage{subcaption}
\usepackage{caption}
\usepackage{float}
\usepackage{tikz}
\usepackage{tikz-cd}
\usepackage{relsize}
\usepackage{pgfplots}

\usepackage{amsthm}
\usepackage{tcolorbox}
\tcbuselibrary{skins}
\tcbuselibrary{breakable}

\makeatletter
\newcommand{\DefineColoredTheoremStyle}[2]{%
  \newtheoremstyle{#1}
    {3pt}{3pt}
    {\itshape}
    {}
    {\bfseries\color{#2}}
    {\;}
    { }
    {\thmname{##1}\thmnumber{ ##2}\textnormal{\thmnote{ (##3)}}}
}
\makeatother

\DefineColoredTheoremStyle{lemmstyle}{black}
\DefineColoredTheoremStyle{propstyle}{black}
\DefineColoredTheoremStyle{thmstyle}{black}
\DefineColoredTheoremStyle{corostyle}{black}
\DefineColoredTheoremStyle{conjstyle}{black}
\DefineColoredTheoremStyle{problstyle}{black}
\DefineColoredTheoremStyle{plainstyle}{black}

\theoremstyle{thmstyle}\newtheorem{theorem}{Theorem}
\theoremstyle{lemmstyle}\newtheorem{lemma}[theorem]{Lemma}
\theoremstyle{propstyle}\newtheorem{prop}[theorem]{Proposition}
\theoremstyle{corostyle}\newtheorem{corollary}[theorem]{Corollary}
\theoremstyle{conjstyle}\newtheorem{conjecture}[theorem]{Conjecture}
\theoremstyle{problstyle}\newtheorem{problem}[theorem]{Problem}
\theoremstyle{plainstyle}\newtheorem{definition}[theorem]{Definition}
\theoremstyle{plainstyle}\newtheorem{assumption}[theorem]{Assumption}
\theoremstyle{plainstyle}\newtheorem{example}[theorem]{Example}
\theoremstyle{plainstyle}\newtheorem{remark}[theorem]{Remark}

\tcbset{
    commonstyle/.style n args={1}{
        oversize=0.5mm,
        colback={#1!8},
        colframe={#1!10},
        boxrule=1pt,
        boxsep=2mm,
        left=1mm,
        right=1mm,
        top=2mm,
        bottom=2mm,
        breakable
    }
}

\newenvironment{boxdefinition}{%
    \begin{tcolorbox}[commonstyle={black}]%
        \begin{definition}%
}{%
        \end{definition}%
    \end{tcolorbox}%
}

\newenvironment{boxassumption}{%
    \begin{tcolorbox}[commonstyle={black}]%
        \begin{assumption}%
}{%
    \end{assumption}%
  \end{tcolorbox}%
}

\newenvironment{boxlemma}{%
  \begin{tcolorbox}[commonstyle={myblue}]%
    \begin{lemma}%
}{%
    \end{lemma}%
  \end{tcolorbox}%
}

\newenvironment{boxproposition}{%
  \begin{tcolorbox}[commonstyle={myorange}]%
    \begin{prop}%
}{%
    \end{prop}%
  \end{tcolorbox}%
}

\newenvironment{boxtheorem}{%
  \begin{tcolorbox}[commonstyle={mymagenta}]%
    \begin{theorem}%
}{%
    \end{theorem}%
  \end{tcolorbox}%
}

\newenvironment{boxcorollary}{%
  \begin{tcolorbox}[commonstyle={myorange}]%
    \begin{corollary}%
}{%
    \end{corollary}%
  \end{tcolorbox}%
}

\newenvironment{boxremark}{%
  \begin{tcolorbox}[commonstyle={mylightgray}]%
    \begin{remark}%
}{%
    \end{remark}%
  \end{tcolorbox}%
}

\usepackage{colortbl}
\usepackage{pifont} 
\usepackage{booktabs}
\usepackage{makecell} 
\usepackage{diagbox}
\usepackage{multirow}

\newcommand{\nc}{\newcommand}
\nc{\rnc}{\renewcommand}

\nc{\<}{\langle}
\rnc{\>}{\rangle}
\nc{\bra}[1]{\langle#1|}
\nc{\ket}[1]{|#1\rangle}
\nc{\ketbra}[2]{|#1\rangle\!\langle#2|}
\nc{\braket}[2]{\langle#1|#2\rangle}
\nc{\braandket}[3]{\langle #1|#2|#3\rangle}
\nc{\proj}[1]{| #1\rangle\!\langle #1 |}
\nc{\avg}[1]{\langle#1\rangle}

\nc{\rank}{\operatorname{Rank}}
\nc{\id}{{\operatorname{id}}}
\nc{\supp}{{\operatorname{supp}}}
\nc{\smfrac}[2]{\mbox{$\frac{#1}{#2}$}}
\nc{\tr}{\operatorname{Tr}}
\nc{\ox}{\otimes}
\nc{\floor}[1]{\lfloor #1 \rfloor}
\nc{\trans}{\mathsf T}
\nc{\img}{\mathbf{i}}
\nc{\Sym}{\operatorname{Sym}}

\def\ve{\varepsilon}

\nc{\cA}{{\cal A}}
\nc{\cB}{{\cal B}}
\nc{\cC}{{\cal C}}
\nc{\cD}{{\cal D}}
\nc{\cE}{{\cal E}}
\nc{\cF}{{\cal F}}
\nc{\cG}{{\cal G}}
\nc{\cH}{{\cal H}}
\nc{\cI}{{\cal I}}
\nc{\cJ}{{\cal J}}
\nc{\cK}{{\cal K}}
\nc{\cL}{{\cal L}}
\nc{\cM}{{\cal M}}
\nc{\cN}{{\cal N}}
\nc{\cO}{{\cal O}}
\nc{\cP}{{\cal P}}
\nc{\cQ}{{\cal Q}}
\nc{\cR}{{\cal R}}
\nc{\cS}{{\cal S}}
\nc{\cT}{{\cal T}}
\nc{\cV}{{\cal V}}
\nc{\cU}{{\cal U}}
\nc{\cX}{{\cal X}}
\nc{\cY}{{\cal Y}}
\nc{\cZ}{{\cal Z}}
\nc{\cW}{{\cal W}}

\nc{\RR}{{{\mathbb R}}}
\nc{\CC}{{{\mathbb C}}}
\nc{\FF}{{{\mathbb F}}}
\nc{\NN}{{{\mathbb N}}}
\nc{\ZZ}{{{\mathbb Z}}}
\nc{\QQ}{{{\mathbb Q}}}
\nc{\UU}{{{\mathbb U}}}
\nc{\EE}{{{\mathbb E}}}

\nc{\bH}{{\mathfrak{H}}}

\nc{\sK}{{{\mathscr{K}}}}
\nc{\sS}{{{\mathscr{S}}}}
\nc{\sT}{{{\mathscr{T}}}}
\nc{\sA}{{{\mathscr{A}}}}
\nc{\sB}{{{\mathscr{B}}}}
\nc{\sC}{{{\mathscr{C}}}}
\nc{\sE}{{{\mathscr{E}}}}
\nc{\sL}{{{\mathscr{L}}}}
\nc{\sG}{{{\mathscr{G}}}}
\nc{\sF}{{{\mathscr{F}}}}
\nc{\sI}{{{\mathscr{I}}}}
\nc{\sN}{{{\mathscr{N}}}}
\nc{\sM}{{{\mathscr{M}}}}

\nc{\Choi}{Choi-Jamio\l{}kowski }
\nc{\reg}{\infty}
\nc{\amo}{\text{\rm amo}}
\nc{\Renyi}{R\'{e}nyi }

\nc{\conv}{\operatorname{conv}}
\nc{\cvxset}{\mathscr{C}}

\nc{\RM}{{{\mathscr{R}}}}

\nc{\END}{\operatorname{End}}
\nc{\PERM}{\mathfrak{\sigma}}

\nc{\Cone}{\text{\rm Cone}}
\nc{\sep}{{\SEP}}

\nc{\DD}{{{\mathbb D}}}
\nc{\BS}{{\scriptscriptstyle \rm {BS}}}
\nc{\Sand}{{\scriptscriptstyle  \rm S}}
\nc{\RSand}{{\scriptscriptstyle  \rm {RS}}}
\nc{\Rev}{{\scriptscriptstyle  \rm {R}}}
\nc{\Petz}{{\scriptscriptstyle  \rm P}}
\nc{\Hypo}{{\scriptscriptstyle  \rm H}}
\nc{\Meas}{{\scriptscriptstyle \rm M}}
\nc{\Proj}{{{\scriptscriptstyle \rm P}}}

\nc{\suchthat}{\text{\rm s.t.}}

\nc{\pl}{{\scalebox{0.7}{+}}}
\nc{\HERM}{\mathscr{H}}
\nc{\PSD}{\HERM_{\pl}}
\nc{\PD}{\HERM_{\pl\pl}}
\nc{\density}{\mathscr{D}}
\nc{\subdensity}{\mathscr{D}_\bullet}

\nc{\polarPSD}[1]{{#1}_{\pl}^{\circ}}
\nc{\polarPSDre}[1]{{#1}_{\pl}^{\star}}
\nc{\polarPD}[1]{{#1}_{\pl\pl}^{\circ}}

\nc{\PPT}{\text{\rm PPT}}
\nc{\Rains}{\text{\rm Rains}}
\nc{\WD}{\text{\rm WD}}
\nc{\SEP}{\text{\rm SEP}}
\nc{\PSEP}{\text{\rm PSEP}}
\nc{\CPTP}{\text{\rm CPTP}}
\nc{\POVM}{\text{\rm POVM}}
\nc{\PVM}{\text{\rm PVM}}
\nc{\CP}{\text{\rm CP}}
\nc{\adv}{\text{\rm adv}}
\nc{\spec}{\text{\rm spec}}
\nc{\poly}{\text{\rm poly}}
\nc{\End}{\operatorname{End}}
\nc{\Par}{\operatorname{Par}}
\nc{\RNG}{\operatorname{RNG}}
\nc{\STAB}{\text{\rm STAB}}
\nc{\epi}{\boldsymbol{\operatorname{epi}}}
\nc{\op}{\boldsymbol{\operatorname{op}}}

\nc{\group}{\mathfrak{G}}
\nc{\perm}{\mathfrak{P}}

\makeatletter
\newcommand*\rel@kern[1]{\kern#1\dimexpr\macc@kerna}
\newcommand*\widebar[1]{%
  \begingroup
  \def\mathaccent##1##2{%
    \rel@kern{0.8}%
    \overline{\rel@kern{-0.8}\macc@nucleus\rel@kern{0.2}}%
    \rel@kern{-0.2}%
  }%
  \macc@depth\@ne
  \let\math@bgroup\@empty \let\math@egroup\macc@set@skewchar
  \mathsurround\z@ \frozen@everymath{\mathgroup\macc@group\relax}%
  \macc@set@skewchar\relax
  \let\mathaccentV\macc@nested@a
  \macc@nested@a\relax111{#1}%
  \endgroup
}
\makeatother

\begin{document}
	\title{\Large \textbf{Operational interpretation of the reverse sandwiched
	\Renyi divergences in composite quantum hypothesis testing}}

	\author[1,2,3]{Masahito Hayashi \thanks{hmasahito@cuhk.edu.cn}}
	\author[1]{Kun Fang \thanks{kunfang@cuhk.edu.cn}}
	\affil[1]{\small School of Data Science, The Chinese University of Hong Kong, Shenzhen,\protect\\  Guangdong, 518172, China}
	\affil[2]{International Quantum Academy, Futian District, Shenzhen 518048, China}
	\affil[3]{Graduate School of Mathematics, Nagoya University, Chikusa-ku, Nagoya 464–8602, Japan}

	\date{\today}

	\maketitle

	\begin{abstract}
		We study the Hoeffding regime of composite quantum hypothesis testing, in which each hypothesis is specified by a sequence of sets of quantum states. We establish quantum Hoeffding bounds under a set of structural assumptions, orthogonal to those of our previous framework. A notable consequence is the direct operational interpretation of the reverse sandwiched \Renyi divergence for $\alpha \in (0,1)$: for the task of discriminating a thermal equilibrium state from a probe state subject to unknown dephasing in the energy eigenbasis, with free Hamiltonian evolution as a special case, the optimal Hoeffding exponent is given exactly by this divergence evaluated on a single copy of the system. The same task in the Stein regime is governed by the reverse quantum relative entropy, providing its operational interpretation as well. This behavior contrasts both with the simple independent and identically distributed (i.i.d.) setting, where the Petz \Renyi divergence and the Umegaki relative entropy govern the Hoeffding and Stein exponents, respectively, and with many composite settings, where only regularized many-copy formulas are available. This finding reveals that passing from simple to composite hypotheses can fundamentally change which quantum divergence determines the operational limits of discrimination, and suggests a new avenue for seeking operational interpretations of quantum divergences by lifting simple hypotheses to richer composite scenarios.
	\end{abstract}

	{ \tableofcontents }

	\section{Introduction}

	Quantum hypothesis testing asks which of two competing models of a physical system is better supported by the outcome of a quantum measurement. It provides the operational backbone for quantifying distinguishability between different quantum sources, with applications across quantum computing, quantum communication, and quantum cryptography~\cite{wilde2011classical,hayashi2017quantum,watrous2018theory}. It also endows formal divergence measures with concrete operational interpretation that turns abstract formulas into tangible information-theoretic quantities and guides their use in practice. Here, an \emph{operational interpretation} means a task-based understanding of a mathematical quantity: just as $\pi$ is the universal ratio between a circle's circumference and its diameter, and Shannon entropy is the minimum average number of bits needed to compress a source in the long run~\cite{shannon1948mathematical}.
	
	In asymmetric hypothesis testing, there are two errors to balance: the Type-I error, which is the probability of incorrectly accepting the alternative hypothesis when the null hypothesis is true, and the Type-II error, which is the probability of incorrectly accepting the null hypothesis when the alternative is true. The central problem is to understand how these two errors trade off as the number of available copies of the system increases. When the hypotheses are simple and given by independent and identically distributed (i.i.d.) copies of a single state, this trade-off admits a particularly clean asymptotic characterization in terms of three complementary regimes, each of which identifies a quantum divergence as the optimal rate for a concrete discrimination task. 
	
	In the Stein regime, one imposes a constant threshold on the Type-I error and asks how quickly the Type-II error can be made to vanish; the optimal exponential rate is given by the Umegaki relative entropy~\cite{hiai1991proper,Ogawa2000}. The identification of this quantity by Hiai and Petz exemplifies the importance of operational interpretations: among the various mathematically plausible quantum extensions of the Kullback--Leibler divergence~\cite{kullback1951information}, including the Belavkin--Staszewski relative entropy~\cite{belavkin1982c}, it is the operational role in hypothesis testing that singles out Umegaki relative entropy as the canonical entropic measure, cementing its central place in quantum information theory. In the strong-converse regime, the Type-II error is forced to decay faster than the relative entropy, and the focus shifts to how quickly the Type-I error approaches unity; here, the sandwiched \Renyi divergence is the relevant quantity~\cite{Ogawa2000,mosonyi2015quantum}. In the Hoeffding regime, the Type-II error is required to decay at a prescribed exponential rate slower than the relative entropy, and the best achievable exponent for the Type-I error is determined by the Petz \Renyi divergence~\cite{nagaoka2006converse,hayashi2007error,audenaert2008asymptotic}. Collectively, these correspondences constitute a central guiding principle of quantum information theory, underpinning plenty of applications~\cite{wilde2011classical,hayashi2017quantum,watrous2018theory} and subsequent generalizations~\cite{cooney2016strong,fang2025towards}.

	In many scenarios of interest, however, the two hypotheses cannot be pinned down to a single pair of states. Adversarial, black-box, and partially characterized models naturally constrain the true state only to belong to some prescribed set, and physical correlations typically preclude a clean description in which every copy of the system is independent and identical. These considerations motivate the study of \emph{composite} and \emph{correlated} hypotheses~\cite{hiai2007large,hiai2008error,brandao2010generalization,berta2021composite,hayashi2016correlation,mosonyi2015two,Mosonyi_2022,fang2024generalized,lami2024solutiongeneralisedquantumsteins,fang2025error,lami2025doubly}, in which the two hypotheses are each described by a whole sequence of sets of quantum states with growing system size, and a single measurement is required to perform uniformly well against every state in each set.

	In a prior work~\cite{fang2025error}, we extended the simple Hoeffding bound along this direction and established a tight quantum Hoeffding bound for hypothesis testing between convex and compact sets of states that are, in addition, \emph{stable under tensor products}, meaning that taking a product of any two states from the set of a given size yields a state that again belongs to the set of the appropriate larger size. This tensor-product stability is structurally central to the argument in that work. It is satisfied in a number of settings of interest, including the set of separable states in entanglement theory and the set of stabilizer states in the resource theory of magic. On the other hand, it fails in some composite scenarios, such as the composite i.i.d.\ setting, in which each hypothesis consists of i.i.d.\ copies of states drawn from a fixed constituent set~\cite{berta2021composite,Mosonyi_2022,lami2025doubly}. 

	In this work, we establish quantum Hoeffding bounds under a different set of structural assumptions that does not require tensor-product stability. Instead, the framework is based on permutation symmetry of the state sequences, together with geometric and topological regularity properties of the associated regularized \Renyi divergences. It is therefore complementary to the framework of~\cite{fang2025error} and is suited to composite scenarios that fall outside the tensor-stable setting.

	Beyond this structural extension, a particular specialization has a notable consequence: it yields an operational interpretation of the \emph{reverse sandwiched \Renyi divergence},
	\begin{align}
		D_{\RSand,\alpha}(\rho\|\sigma) := \frac{1}{\alpha-1}\log \tr \left[\left(\rho^{\frac{\alpha}{2(1-\alpha)}}\sigma\rho^{\frac{\alpha}{2(1-\alpha)}}\right)^{1-\alpha}\right],
	\end{align}
	with $\alpha \in (0,1)$. This is a quantum \Renyi divergence that extends the well-studied sandwiched \Renyi divergence from its usual parameter range $\alpha\geq 1/2$ to the complementary range $\alpha \in (0,1)$~\cite{audenaert2015alpha}. Although it has been used as a quantitative tool in~\cite{lipka2024quantum,warsi2025generalization}, its operational interpretation has remained open since its introduction in~\cite{audenaert2015alpha}. In the limit $\alpha \to 1$, the reverse sandwiched \Renyi divergence converges to a \emph{reverse quantum relative entropy} $D_{\Rev}(\rho\|\sigma):= \lim_{\alpha\to 1} D_{\RSand,\alpha}(\rho\|\sigma)$~\cite{audenaert2015alpha}, which has been studied quantitatively in~\cite{lipka2024quantum,Hayashi2025another}. We show that both $D_{\RSand,\alpha}$ and $D_{\Rev}$ acquire direct operational meanings as the optimal error exponents of a composite hypothesis-testing problem that arises naturally in quantum thermodynamics.
	
	Specifically, consider a quantum system with Hamiltonian $H=\sum_j E_j\ket{E_j}\bra{E_j}$ and energy eigenbasis $\{\ket{E_j}\}$. At inverse temperature $\beta$, its thermal equilibrium state is
	$\rho=e^{-\beta H}/\tr[e^{-\beta H}]$. A fundamental task in quantum thermodynamics is to decide whether the system is in this thermal equilibrium state $\rho$ or in a prescribed non-equilibrium probe state $\sigma$~\cite{horodecki2013fundamental,brandao2013resource}. In realistic implementations, however, the state reaching the tester may be affected by unknown noise or imperfect control~\cite{watanabe2024black,zhang2026quantum}. 
	One particular model of this uncertainty is phase noise in the energy eigenbasis, arising for example from imprecise waiting times, clock misalignment, or energy fluctuations. The relevant phase rotations form the diagonal unitary family
	\begin{align}
		\group:=\left\{\sum_j e^{i\theta_j}\ket{E_j}\bra{E_j}:\theta_j\in[0,2\pi)\right\}.
	\end{align}
	Thus the tester may know that the intended states to distinguish are $\rho$ or $\sigma$ but not which $g\in\group$ has acted on them. If the same unknown phase setting applies throughout an $n$-copy experiment, the two hypotheses become the sets of all states generated from $\rho^{\otimes n}$ and $\sigma^{\otimes n}$ by these diagonal phase rotations. This leads to the following composite hypothesis-testing problem:
	\begin{align}
		\text{(Null)} \quad  \sA_n := \{\rho^{\otimes n}\}, \qquad
		\text{(Alternative)} \quad \sB_n := \left\{g^{\otimes n} \sigma^{\otimes n} (g^\dagger)^{\otimes n} : g \in \group \right\},
	\end{align}
	where we note that $\rho$ is invariant under the action of $\group$ by definition.
	Let $\alpha_{n,r}(\sA_n\|\sB_n)$ denote the optimal worst-case Type-I error, given that the worst-case Type-II error decays exponentially at rate $r$. Theorem~\ref{thm: RSand operational meaning} establishes the following Hoeffding-type formula:
	\begin{align}
		\lim_{n\to\infty} -\frac{1}{n} \log \alpha_{n,r}(\sA_n\|\sB_n)
		= \sup_{\alpha\in(0,1)} \frac{\alpha-1}{\alpha} \bigl(r - D_{\RSand,\alpha}(\rho\|\sigma)\bigr),
	\end{align}
	for any $0<r<D_{\Rev}(\rho\|\sigma)$.
	Thus, for this composite discrimination problem, the optimal Hoeffding exponent is determined by the single-letter reverse sandwiched \Renyi divergence between the thermal and probe states. Moreover, conventional time evolution under the Hamiltonian $H$ is included as a one-parameter subgroup of $\group$ with $\theta_j=-E_jt$; under a generic density condition stated in Corollary~\ref{coro: RSand time evolution}, this smaller time-evolution family yields the same exponent.
	Similarly, in the Stein regime, let $\beta_{n,\ve}(\sA_n\|\sB_n)$ denote the optimal Type-II error exponent with a constant threshold $\ve$ on the Type-I error. We show in Corollary~\ref{coro: RSand operational meaning Stein} that 
	\begin{align}
		\lim_{n\to\infty} -\frac{1}{n} \log \beta_{n,\ve}(\sA_n\|\sB_n) = D_{\Rev}(\rho\|\sigma).
	\end{align}

	These results are conceptually significant from two perspectives. First, the governing quantity differs from the Petz \Renyi divergence that characterizes the Hoeffding exponent for simple i.i.d.\ hypotheses~\cite{nagaoka2006converse,hayashi2007error,audenaert2008asymptotic}; to the best of our knowledge, this is the first setting in which the reverse sandwiched divergence arises as the \emph{exact} rate of a Hoeffding-type exponent. It shows that moving from simple to composite hypotheses is not merely a formal generalization: it can genuinely change the divergence that governs the optimal asymptotics, suggesting that our prior experience may fail in this richer landscape. Second, although error exponents in composite hypothesis testing are typically given by regularized many-copy formulas that are difficult to evaluate~\cite{berta2021composite,Mosonyi_2022,lami2025doubly}, the exponent obtained here admits an explicit single-letter expression. This points to a useful methodological lesson for the search for operational meanings of quantum divergences: one need not restrict attention to simple hypotheses. Instead, it can be fruitful to lift the problem to a richer composite setting, where the enlarged hypothesis space may reveal  structures that remain invisible in the simpler formulation, much as lifting methods in optimization expose useful structure by embedding a problem in a higher-dimensional space~\cite{fawzi2022lifting}.

	\paragraph{Organization.} The rest of the paper is organized as follows. Section~\ref{sec: preliminaries} introduces the notation, quantum divergences, and hypothesis-testing framework used in this work. Section~\ref{sec: Hoeffding bounds} presents the main technical results on quantum Hoeffding bounds for composite hypotheses. Section~\ref{sec: matching bounds} applies these bounds to the problem of discriminating a thermal state from a dephased probe state and establishes the operational interpretation of the reverse sandwiched \Renyi divergence and the reverse quantum relative entropy. Finally, Section~\ref{sec: discussion} concludes with a summary and outlook.

	\section{Preliminaries}
	\label{sec: preliminaries}

	\subsection{Notation}

	Throughout this work, we adopt the following notational conventions. Finite-dimensional
	Hilbert spaces are denoted by $\cH$, with $|\cH|$ indicating their dimension. The
	set of all linear operators on $\cH$ is denoted by $\sL(\cH)$, while
	$\HERM(\cH)$ and $\HERM_{\pl}(\cH)$ denote the sets of Hermitian and positive
	semidefinite operators on $\cH$, respectively. The set of density operators (i.e.,
	positive semidefinite operators with unit trace) on $\cH$ is denoted by
	$\density(\cH)$. Calligraphic letters such as $\sA$, $\sB$, and $\sC$ are used
	to represent sets of linear operators or sequences of such sets. Unless otherwise specified, all logarithms
	are taken to base two and denoted by $\log(x)$. Throughout, we adopt the convention $0^{x} := 0$ for all $x \in \RR$ and $\log 0 = -\infty$. Denote $\rho \ll \sigma$ if
	the support of $\rho$ is contained within the support of $\sigma$. The positive
	semidefinite ordering is written as $X \geq Y$ if and only if $X - Y \geq 0$.
	The absolute value of an operator $X$ is defined as
	$|X|:= (X^{\dagger} X)^{1/2}$. For a Hermitian operator $X$ with spectral decomposition
	$X = \sum_{i} x_{i} E_{i}$, the projection onto the non-negative eigenspaces
	is denoted by $\{X \geq 0\} := \sum_{x_i \geq 0}E_{i}$. Similarly,
	$\{X > 0\} := \sum_{x_i > 0}E_{i}$. 

	A real-valued function $f$ on a convex set $C$ is said to be \emph{strictly convex} on $C$ if
	\begin{align}
		f((1-\lambda)x_{1} + \lambda x_{2}) < (1-\lambda)f(x_{1}) + \lambda f(x_{2}), \qquad \forall 0 < \lambda < 1,
	\end{align}
	for any two different points $x_{1}$ and $x_{2}$ in $C$.

	\subsection{Quantum divergences}

	A functional $\DD: \density \times \PSD \to \RR$ is called a \emph{quantum
	divergence} if it satisfies the data-processing inequality: for any completely
	positive and trace-preserving (CPTP) map $\cE$ and any $(\rho,\sigma) \in \density
	\times \PSD$, it holds that
	$\DD(\cE(\rho)\|\cE(\sigma)) \leq \DD(\rho\|\sigma)$. In the following, we
	introduce several quantum divergences that will be used throughout this work. We
	also define quantum divergences between two sets of quantum states. 

	\begin{definition}
		For any $\rho\in \density$ and $\sigma \in \PSD$, the Umegaki relative
		entropy is defined by~\cite{umegaki1962conditional}
		\begin{align}
			\label{eq: Umegaki}
			D(\rho\|\sigma):= \tr\!\left[\rho(\log \rho - \log \sigma)\right],
		\end{align}
		if $\rho \ll \sigma$, and $+\infty$ otherwise.
	\end{definition}

	\begin{definition}
		For any $\rho\in \density$ and $\sigma \in \PSD$, the Petz \Renyi divergence
		is defined by~\cite{petz1986quasi}:
		\begin{align}
			D_{\Petz,\alpha}(\rho\|\sigma) := \frac{1}{\alpha-1}\log Q_{\Petz,\alpha}(\rho\|\sigma),
		\end{align}
		with the Petz \Renyi quasi-divergence defined by
		\begin{align}
			\label{eq: Petz-quasi}Q_{\Petz,\alpha}(\rho\|\sigma) := \tr\left[\rho^{\alpha}\sigma^{1-\alpha}\right]
		\end{align}
		if $\alpha \in (0,1)$, or if $\alpha > 1$ and $\rho \ll \sigma$, and $Q_{\Petz,\alpha}(\rho\|\sigma) := +\infty$ otherwise.
	\end{definition}

	\begin{definition}
		For any $\rho\in \density$, $\sigma \in \PSD$, the sandwiched \Renyi divergence
		is defined by~\cite{muller2013quantum,wilde2014strong}:
		\begin{align}
			D_{\Sand,\alpha}(\rho\|\sigma) := \frac{1}{\alpha-1}\log Q_{\Sand,\alpha}(\rho\|\sigma),
		\end{align}
		with the sandwiched \Renyi quasi-divergence defined by
		\begin{align}
			\label{eq: sandwiched-quasi}Q_{\Sand,\alpha}(\rho\|\sigma) := \tr\left[\left(\sigma^{\frac{1-\alpha}{2\alpha}}\rho\sigma^{\frac{1-\alpha}{2\alpha}}\right)^{\alpha}\right]
		\end{align}
		if $\alpha \in (0,1)$, or if $\alpha > 1$ and $\rho \ll \sigma$, or if $\alpha < 0$ and $\sigma \ll \rho$, and $Q_{\Sand,\alpha}(\rho\|\sigma) := +\infty$ otherwise.
	\end{definition}

	\begin{boxdefinition}
		[Reverse sandwiched \Renyi divergence.]\label{def: reverse sandwiched Renyi divergence}
		Let $\alpha \in (0,1) \cup (1, +\infty)$. For any $\rho,\sigma\in \density$, the reverse sandwiched
		\Renyi divergence is defined by~\cite{audenaert2015alpha}:
		\begin{align}
			\label{eq: sandwiched-R}D_{\RSand,\alpha}(\rho\|\sigma) := \frac{\alpha}{1-\alpha}D_{\Sand,1-\alpha}(\sigma\|\rho).
		\end{align}
		The corresponding reverse quantum relative entropy is defined by~\cite{audenaert2015alpha}:
		\begin{align}
			D_{\Rev}(\rho\|\sigma):= \lim_{\alpha\to 1}D_{\RSand,\alpha}(\rho\|\sigma) .
		\end{align}
	\end{boxdefinition}

	Note that $D_{\RSand,\alpha}$ satisfies the data-processing inequality for $\alpha \in (0,1/2)$ but not for $\alpha \in (1/2,1)$. A closed-form expression for $D_{\Rev}(\rho\|\sigma)$ is given in~\cite[Theorem~2]{audenaert2015alpha}; in general $D_{\Rev}(\rho\|\sigma) \le D(\rho\|\sigma)$, with equality for commuting states.

	\begin{example}
		\label{ex: divergence comparison} Figure~\ref{fig: divergence comparison} illustrates
		the ordering among the Petz, sandwiched, and reverse sandwiched \Renyi divergences
		for a non-commuting qubit pair:
		\begin{align}
			\rho = \begin{pmatrix}0.8&0 \\ 0&0.2\end{pmatrix}, \qquad \sigma = \begin{pmatrix}0.7&0.3 \\ 0.3&0.3\end{pmatrix}.
		\end{align}
		The plot shows that for $\alpha\in(0,1)$, the sandwiched divergence is dominated
		by the Petz divergence, which in turn is bounded above by the Umegaki
		relative entropy, i.e., $D_{\Sand,\alpha}(\rho\|\sigma) \le D_{\Petz,\alpha}(
		\rho\|\sigma) \le D(\rho\|\sigma)$. The sandwiched and reverse sandwiched
		divergences coincide at $\alpha = 1/2$, while for $\alpha > 1/2$ the reverse
		sandwiched divergence is strictly smaller than the sandwiched one. As
		$\alpha\to 1^-$, both $D_{\Petz,\alpha}$ and $D_{\Sand,\alpha}$ converge to the
		Umegaki relative entropy $D(\rho\|\sigma)$, whereas the reverse sandwiched
		divergence $D_{\RSand,\alpha}$ converges to the strictly smaller limit
		$D_{\Rev}(\rho\|\sigma)$.
	\end{example}

	\begin{figure}[ht]
		\centering
		\includegraphics[width=0.7\textwidth]{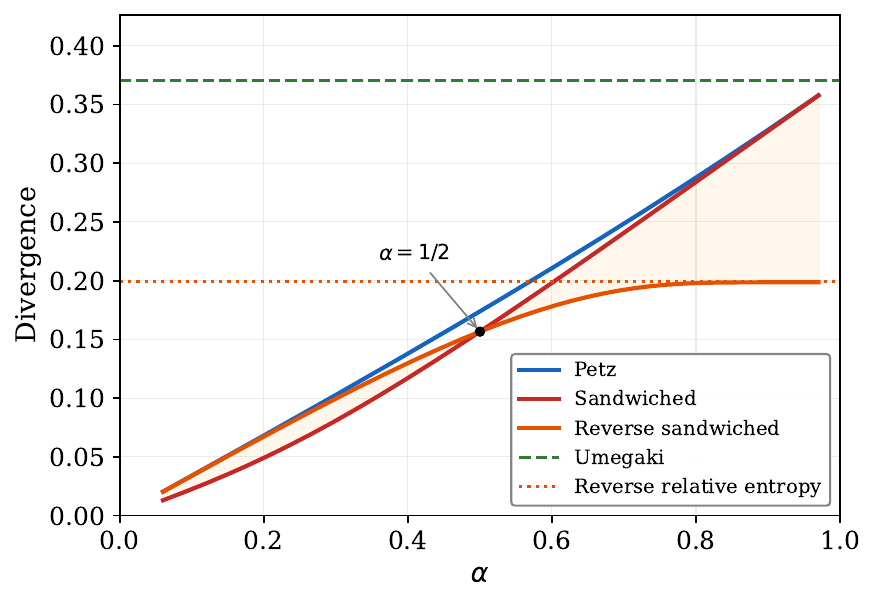}
		\caption{Comparison of the Petz, sandwiched, and reverse sandwiched \Renyi
		divergences as functions of $\alpha\in(0,1)$ given in Example~\ref{ex: divergence
		comparison}. The dashed line marks the Umegaki relative entropy $D(\rho\|\sigma)$, and the dotted line marks the reverse quantum relative entropy $D_{\Rev}(\rho\|\sigma)$. The shaded
		region highlights the gap between sandwiched and reverse sandwiched
		divergences.}
		\label{fig: divergence comparison}
	\end{figure}

	\begin{definition}[Quantum divergence between two sets of states.]
		\label{def: divergence between two sets} Let $\DD$ be a quantum divergence between
		two quantum states. Let $\cH$ be a finite-dimensional Hilbert space. Then for
		any sets $\sA,\sB\subseteq \density(\cH)$, the quantum divergence between
		these two sets is defined by
		\begin{align}
			\DD(\sA\|\sB):= \inf_{\substack{\rho \in \sA\\ \sigma \in \sB}}\DD(\rho\|\sigma).
		\end{align}
		Let $\sA = \{\sA_{n}\}_{n\in \NN}$ and $\sB = \{\sB_{n}\}_{n\in \NN}$ be two
		sequences of sets of quantum states\footnote{We abuse the notation $\sA,\sB$
		to refer both to sets of states and to sequences of such sets, depending on the
		context.}, where each $\sA_{n}, \sB_{n} \subseteq \density(\cH^{\ox n})$. The
		regularized divergence between these sequences is defined by
		\begin{align}
			\underline{\DD}^{\reg}(\sA \| \sB) & := \liminf_{n \to \infty}\frac{1}{n}\DD(\sA_{n}\| \sB_{n}), \\
			\overline{\DD}^{\reg}(\sA \| \sB)  & := \limsup_{n \to \infty}\frac{1}{n}\DD(\sA_{n}\| \sB_{n}).
		\end{align}
		If the limit exists, we define the regularized divergence as
		\begin{align}
			\DD^{\reg}(\sA \| \sB) := \lim_{n \to \infty}\frac{1}{n}\DD(\sA_{n}\| \sB_{n}).
		\end{align}
	\end{definition}

	\subsection{Hypothesis testing between two sets of quantum states}
	\label{sec: Hypothesis testing between two sets of quantum states}

	As in standard hypothesis testing, two types of errors can occur: a Type-I error,
	where a sample from $\sA_{n}$ is incorrectly classified as coming from
	$\sB_{n}$, and a Type-II error, where a sample from $\sB_{n}$ is incorrectly
	classified as coming from $\sA_{n}$. Since we aim to control the errors for any
	state within the sets, regardless of which one is drawn, the (worst-case)
	Type-I error is defined by
	\begin{align}\label{eq: worst type I}
		\alpha(\sA_{n}, M_{n}): = \sup_{\rho_n \in \sA_n}\tr[\rho_{n} (I-M_{n})],
	\end{align}
	and the (worst-case) Type-II error is defined by
	\begin{align}\label{eq: worst type II}
		\beta(\sB_{n}, M_{n}): = \sup_{\sigma_n \in \sB_n}\tr[\sigma_{n} M_{n}].
	\end{align}

	The Hoeffding regime studies the optimal
	behavior of the Type-I error provided that the Type-II error exponentially
	decays. More explicitly, the optimal Type-I error for hypothesis testing
	between two sets of quantum states, $\sA_{n}$ and $\sB_{n}$, is defined as
	\begin{align}
		\label{eq: optimal Type-I error sets 2}\alpha_{n, r}(\sA_{n}\|\sB_{n}):= \min_{0\leq M_n \leq I}\left\{\alpha(\sA_{n}, M_{n}): \beta(\sB_{n}, M_{n}) \leq 2^{-nr}\right\},
	\end{align}
	where the measurement $M_{n}$ is chosen to minimize the worst-case Type-I
	error $\alpha(\sA_{n}, M_{n})$, subject to the constraint that the Type-II error
	$\beta(\sB_{n}, M_{n})$ decays exponentially at a rate $r$. In other words, the
	measurement must perform universally well for all states in $\sA_{n}$ and $\sB_{n}$.

	Similarly, the Stein regime studies the optimal behavior of the Type-II error provided that the Type-I error is bounded by a constant threshold $\ve \in (0,1)$. More explicitly, the optimal Type-II error for hypothesis testing between two sets of quantum states, $\sA_{n}$ and $\sB_{n}$, is defined as
	\begin{align}
		\label{eq: optimal Type-II error sets 2}\beta_{n, \ve}(\sA_{n}\|\sB_{n}):= \min_{0\leq M_n \leq I}\left\{\beta(\sB_{n}, M_{n}): \alpha(\sA_{n}, M_{n}) \leq \ve\right\}.
	\end{align}

	The following result is a useful technical tool that allows us to reduce the discrimination between convex sets of quantum states to the discrimination between individual states.

	\begin{lemma}[{\cite[Lemma 24]{fang2025error}}.]
		\label{lem: optimal Type-I error minimax} Let $\cH$ be a finite-dimensional Hilbert
		space, $r >0$ be a real number, and $n\in \NN$. Let
		$\sA_{n}, \sB_{n}\subseteq \density(\cH^{\ox n})$ be two convex sets. Then it
		holds that
		\begin{align}
			\alpha_{n, r}(\sA_{n}\| \sB_{n}) = \sup_{\substack{\rho_n \in \sA_n\\ \sigma_n \in \sB_n}}\alpha_{n, r}(\rho_{n}\| \sigma_{n}).
		\end{align}
	\end{lemma}

	\section{Quantum Hoeffding bounds}
	\label{sec: Hoeffding bounds}

	This section presents the main technical contribution of the paper. We begin by establishing a general quantum Hoeffding lower bound, expressed in terms of the regularized reverse sandwiched \Renyi divergence, in Theorem~\ref{thm: quantum Hoeffding lower bound} of Section~\ref{sec: Quantum Hoeffding lower bound}:
	\begin{align}
		\liminf_{n\to \infty}-\frac{1}{n}\log \alpha_{n, r}(\sA_{n}\| \sB_{n}) & \geq \sup_{\alpha \in (0,1)}\frac{\alpha-1}{\alpha}\bigg(r - \underline{D}_{\RSand,\alpha}^{\infty}(\sA\|\sB)\bigg).
	\end{align}
	We then establish a complementary upper bound in terms of the regularized Petz \Renyi divergence, in Theorem~\ref{thm: quantum Hoeffding upper bound} of Section~\ref{sec: quantum Hoeffding upper bound}:
	\begin{align}
		\limsup_{n\to \infty} -\frac{1}{n} \log \alpha_{n,r}(\sA_n\|\sB_n) &\leq \sup_{\alpha \in (0,1)} \frac{\alpha-1}{\alpha} \bigg(r - {D}_{\Petz,\alpha}^{\infty}(\sA\| \sB) \bigg).
	\end{align}
	While these bounds are not necessarily tight in general, we will show in Section~\ref{sec: matching bounds} that they coincide in a specific composite setting, where the reverse sandwiched \Renyi divergence acquires a direct operational interpretation.

	For the ease of presentation and comparison, we list the following assumptions.

	\begin{boxassumption} 
		Let $\cH$ be a finite-dimensional Hilbert space. For each $n \in \NN$, let $\sC_n, \sC'_n \subseteq \density(\cH^{\otimes n})$, and define the sequences $\sC = \{\sC_n\}_{n \in \NN}$ and $\sC' = \{\sC'_n\}_{n \in \NN}$. Let $\perm_n$ denote the permutation group on $n$ elements, and $U_{\pi}$ the natural unitary representation of $\pi \in \perm_n$ on $\cH^{\otimes n}$.
		\begin{center}
			\renewcommand{\arraystretch}{1.5}
			\begin{tabular}{@{} c l p{10.5cm} @{}}
				\toprule \textbf{Label} & \textbf{Name}             & \textbf{Description}                                                                                                                                              \\
				\midrule (C1)           & Convexity                 & For any $n \in \NN$, the set $\sC_{n}$ is convex.                                                                                                                 \\
				(C2)                    & Compactness               & For any $n \in \NN$, the set $\sC_{n}$ is compact.                                                                                                                \\
				(C3)                    & Stability                 & For any $m,n \in \NN$, $\rho_m \ox \sigma_{n}\in \sC_{m+n}, \forall \rho_m \in \sC_m, \forall \sigma_n \in \sC_n$.                                                                                                   \\
				(C4)                    & Finiteness                & $D_{\Petz, \alpha}(\sC_{1}\|\sC'_{1}) < \infty$ for $\alpha \in (0,1)$.                                                                       \\
				(C5)                   & $\perm$-closedness & For any $n \in \NN$, $U_\pi \rho_{n} U_\pi^{\dagger} \in \sC_{n}$, $\forall \rho_{n} \in \sC_{n}, \forall \pi \in \perm_n$.                                 \\
				(C5$'$)                    & $\perm$-invariance & For any $n \in \NN$, $U_\pi \rho_{n} U_\pi^{\dagger} = \rho_{n}$, $\forall \rho_{n} \in \sC_{n}, \forall \pi \in \perm_n$.                                  \\
				(C6) & Strict-concavity & $\alpha \mapsto (1-\alpha)D_{\Petz,\alpha}^{\infty}(\sC\|\sC')$ is strictly concave on $(0,1)$.\\
				(C7) & Continuity & $\alpha \mapsto (1-\alpha)D_{\Petz,\alpha}^{\infty}(\sC\|\sC')$ is $C^1$-continuous on $(0,1)$.  \\
				\bottomrule
			\end{tabular}
		\end{center}
	\end{boxassumption}

	The following Table~\ref{tab:assumptions-bounds} summarizes how the assumptions in this paper differ from those in~\cite{fang2025error}. The earlier results rely on compactness, stability, and finiteness (C2)--(C4). By contrast, the lower bound here is driven by permutation symmetry, and the upper bound by geometric and topological regularity of the regularized \Renyi divergence. In particular, neither bound in the present work requires tensor stability.

	\begin{table}[ht]
		\centering
		\renewcommand{\arraystretch}{1.5}
		\setlength{\tabcolsep}{5pt}

		\newcolumntype{P}{>{\columncolor{myblue!10}}c}
		\newcolumntype{H}{>{\columncolor{mymagenta!10}}c}

		\resizebox{\textwidth}{!}{%
		\begin{tabular}{l P P P P | H H H H}
			\toprule
			\multirow{3}{*}{\textbf{Assumption}}
			 & \multicolumn{4}{>{\columncolor{myblue!10}}c|}{\textbf{Previous work}~\cite{fang2025error}}
			 & \multicolumn{4}{>{\columncolor{mymagenta!10}}c}{\textbf{This work}}                                                         \\
			\cmidrule(lr){2-5}\cmidrule(lr){6-9}
			 & \multicolumn{2}{>{\columncolor{myblue!10}}c}{Lower bound}
			 & \multicolumn{2}{>{\columncolor{myblue!10}}c|}{Upper bound}
			 & \multicolumn{2}{>{\columncolor{mymagenta!10}}c}{Lower bound (Thm.~\ref{thm: quantum Hoeffding lower bound})}
			 & \multicolumn{2}{>{\columncolor{mymagenta!10}}c}{Upper bound (Thm.~\ref{thm: quantum Hoeffding upper bound})}                 \\
			\cmidrule(lr){2-3}\cmidrule(lr){4-5}\cmidrule(lr){6-7}\cmidrule(lr){8-9}
			 & $\sA$      & $\sB$      & $\sA$      & $\sB$      & $\sA$      & $\sB$      & $\sA$      & $\sB$      \\
			\midrule
			(C1) Convexity                    & \checkmark & \checkmark &            &            &    \checkmark        &  & \checkmark & \checkmark \\
			(C2) Compactness                  & \checkmark & \checkmark &            &            &            &            &            &            \\
			(C3) Stability                    &            &            & \checkmark & \checkmark &            &            &            &            \\
			(C4) Finiteness                   &            &            & \checkmark & \checkmark &            &            &            &            \\
			\hline
			(C5) $\perm$-closedness &            &            &            &            &        \checkmark    &  &            &            \\
			(C5$'$) $\perm$-invariance    &            &            &            &            &  &    \checkmark        &            &            \\
	
			(C6) Strict-concavity             &            &            &            &            &            &            & \checkmark & \checkmark \\
			(C7) Continuity                   &            &            &            &            &            &            & \checkmark & \checkmark \\
			\bottomrule
		\end{tabular}%
		}
		\caption{Comparison of the structural assumptions imposed on the state sequences $\sA$ and $\sB$ by the Hoeffding bounds of~\cite{fang2025error} (left, blue) and those established in this work (right, red). The two sets of assumptions are substantially different: the previous bounds rely on compactness, stability, and finiteness (C2)--(C4), whereas the bounds obtained here replace these by the permutation-symmetry conditions (C5)--(C5$'$) for the lower bound and by geometric and topological conditions (C6)--(C7) on the R\'enyi divergence for the upper bound.}
		\label{tab:assumptions-bounds}
	\end{table}

	\subsection{Quantum Hoeffding lower bound}
	\label{sec: Quantum Hoeffding lower bound}

	This section establishes the quantum Hoeffding lower bound. We first develop three technical lemmas that constitute the core of the proof. Lemma~\ref{lem: reduction perm inv} exploits the permutation symmetry of $\sA_n$ and $\sB_n$ to reduce the Type-II optimization to permutation-invariant states, thereby enabling the use of pinching techniques. Lemma~\ref{lem: blockwise NP tests} then constructs a blockwise Neyman--Pearson test for each pinched state, and Lemma~\ref{lem: universal test pinching} assembles these blockwise tests into a single universal test that controls the error uniformly over all states in $\sA_n$ and $\sB_n$. With these ingredients, we state and prove the main result in two forms: Proposition~\ref{prop: quantum Hoeffding lower bound} expresses the bound via the sandwiched \Renyi divergence, while Theorem~\ref{thm: quantum Hoeffding lower bound} reformulates it to the reverse sandwiched \Renyi divergence.

	Let $\perm_n$ be the symmetric group on $n$ elements
	and $U_{\pi}$ denotes the natural unitary representation of $\pi$ on
		$\cH^{\ox n}$. Define the twirling map
		\begin{align}
			\cT_{\perm}(X)\coloneqq \frac{1}{n!}\sum_{\pi\in \perm_n}U_{\pi} X U_{\pi}^{\dagger},
		\end{align}
		and the resulting set of twirled states from $\sB_n$,
		\begin{align}
			\cT_{\perm}(\sB_{n}):= \{\cT_{\perm}(\sigma_{n}): \sigma_{n} \in \sB_{n}\}.
		\end{align}

	The following result is similar in spirit to~\cite[Lemma 2]{hayashi2025general}, which exploits the symmetry of the sets to restrict the optimization to symmetric states, but it applies to the worst-case Type-I error and requires a different set of assumptions.
	\begin{boxlemma}
		[Reduction to permutation-invariant states.]\label{lem: reduction perm inv}
		Let $\cH$ be a finite-dimensional Hilbert space and
		$\sA_{n}, \sB_{n} \subseteq \density(\cH^{\ox n})$ be two sets of quantum
		states. Suppose that 
		\begin{itemize}
			\item $\sA_{n}$ satisfies~(C5$'$) $\perm$-invariance;
			\item $\sB_{n}$ satisfies~(C1) convexity and~(C5) $\perm$-closedness.
		\end{itemize}
		
		Then for any $r>0$, it holds that
		\begin{align}
			\alpha_{n,r}(\sA_{n}\|\sB_{n}) = \sup_{\sigma_n\in\cT_{\perm}(\sB_{n})}\alpha_{n,r}(\sA_{n}\|\sigma_{n}).
		\end{align}
	\end{boxlemma}

	\begin{proof}
		For any $\sigma_{n}\in\sB_{n}$, we first show that
		\begin{align}
			\label{eq: alpha nr twirling inequality}\alpha_{n,r}(\sA_{n}\|\sigma_{n})\le \alpha_{n,r}(\sA_{n}\|\cT_{\perm}(\sigma_{n})).
		\end{align}
		To verify this, let $M_{n}$ be any feasible test for $\alpha_{n,r}(\sA_{n}\|\cT_{\perm}
		(\sigma_{n}))$, satisfying $0\le M_{n} \le I$ and
		$\tr[\cT_{\perm}(\sigma_{n})\,M_{n}]\le 2^{-nr}$. Since $\cT_{\perm}$ is CPTP, its adjoint
		$\cT_{\perm}^{\dagger}$ is unital and completely positive. Define the test ${M}'_{n}
		:= \cT_{\perm}^{\dagger}(M_{n})$, which satisfies $0\leq{M}'_{n} \leq I$. Its Type-II
		error against $\sigma_{n}$ is
		\begin{align}
			\tr[\sigma_{n}\,{M}'_{n}] = \tr[\sigma_{n}\,\cT_{\perm}^{\dagger}(M_{n})]=\tr[\cT_{\perm}(\sigma_{n})\,M_{n}]\le 2^{-nr},
		\end{align}
		so ${M}'_{n}$ is feasible for the optimization defining $\alpha_{n,r}(\sA_{n}
		\|\sigma_{n})$. Therefore,
		\begin{align}
			\alpha_{n,r}(\sA_{n}\|\sigma_{n}) & \le \sup_{\rho_n\in\sA_n}\tr\!\left[\rho_{n}\bigl(I-{M}'_{n}\bigr)\right]  \\
			& = \sup_{\rho_n\in\sA_n}\tr\!\left[\rho_{n}\bigl(I-\cT_{\perm}^{\dagger}(M_{n})\bigr)\right]                     \\
			& = \sup_{\rho_n\in\sA_n}\tr\!\left[\cT_{\perm}(\rho_{n})(I-M_{n})\right]                                         \\
			& \le \sup_{\rho_n\in\sA_n}\max_{\pi\in \perm_n}\tr\!\left[U_{\pi}\rho_{n} U_{\pi}^{\dagger} (I-M_{n})\right] \\
			& = \sup_{\rho_n\in\sA_n}\tr\!\left[\rho_{n}(I-M_{n})\right],
		\end{align}
		where the second inequality relaxes the average value to the maximum, and the last equality uses the permutation invariance~(C5$'$) of $\sA_{n}$.
		Since this holds for every feasible test $M_{n}$ for
		$\alpha_{n,r}(\sA_{n}\|\cT_{\perm}(\sigma_{n}))$, we conclude~\eqref{eq: alpha nr twirling inequality}.

		Moreover, since $\cT_{\perm}(\sigma_{n})$ is a convex combination of unitarily
		conjugated copies of $\sigma_{n}$, the convexity~(C1) and permutation closedness~(C5)
		of $\sB_{n}$ guarantee that $\cT_{\perm}(\sigma_{n})\in\sB_{n}$. Therefore,
		\begin{align}
			\sup_{\sigma_n \in \sB_n}\alpha_{n,r}(\sA_{n}\|\sigma_{n}) = \sup_{\sigma_n \in \cT_{\perm}(\sB_{n})}\alpha_{n,r}(\sA_{n}\|\sigma_{n}).\label{eq: perm inv tmp1}
		\end{align}

		Finally, we have
		\begin{align}
			\alpha_{n,r}(\sA_{n}\|\sB_{n}) & = \alpha_{n,r}(\conv(\sA_{n})\|\sB_{n})                                                        \\ & = \sup_{\substack{\rho_n\in\conv(\sA_n)\\ \sigma_n\in\sB_n}}\alpha_{n,r}(\rho_{n}\|\sigma_{n}) \\
			& = \sup_{\sigma_n\in\sB_n}\alpha_{n,r}(\sA_{n}\|\sigma_{n})                                     \\
			& = \sup_{\sigma_n\in\cT_{\perm}(\sB_{n})}\alpha_{n,r}(\sA_{n}\|\sigma_{n}),
		\end{align}
		where the first and third equalities follow from the linearity of the Type-I
		and Type-II errors, the second from Lemma~\ref{lem: optimal Type-I error minimax}
		together with the convexity~(C1) of $\sB_{n}$, and the fourth from~\eqref{eq: perm inv tmp1}.
	\end{proof}

	Let $\cH$ be a finite-dimensional Hilbert space with $d := \dim\cH$, and
	let $n \in \NN$. We denote by $\Lambda_{n}$ the set of Young diagrams
	(equivalently, partitions of $n$ into at most $d$ parts) that label the
	irreducible representations in the Schur--Weyl decomposition
	of $\cH^{\otimes n}$:
	\begin{align}
		\cH^{\otimes n}=\bigoplus_{\lambda \in \Lambda_{n}}\cU_{\lambda} \otimes \cV_{\lambda},
	\end{align}
	where $\cU_{\lambda}$ and $\cV_{\lambda}$ are the irreducible representation
	spaces of the unitary group $\mathfrak{U}(\cH)$ and the symmetric group $\perm_n$,
	respectively. Write
	\begin{align}
		d_{\lambda}:=\dim \cU_{\lambda} \quad \text{and}\quad d_{n}:=\sum_{\lambda\in\Lambda_n}d_{\lambda}.
	\end{align}
	Then every permutation-invariant state $\omega_{n} \in \density(\cH^{\ox n})$
	is block-diagonal in this decomposition:
	\begin{align}
		\omega_{n}=\bigoplus_{\lambda\in\Lambda_n}\omega_{\lambda}(\omega_{n})\otimes \pi_{\cV_\lambda},
	\end{align}
	where $\pi_{\cV_\lambda}$ denotes the maximally mixed state on $\cV_{\lambda}$.

	\begin{boxlemma}
		[Blockwise Neyman--Pearson tests for pinched state.]\label{lem: blockwise NP tests}

		Let $\cH$ be a finite-dimensional Hilbert space, $n \in \NN$, and $\alpha \in (0,1)$. Let $\sA_n \subseteq \density(\cH^{\otimes n})$ be a set of quantum states, and let $\sigma_n \in \density(\cH^{\otimes n})$ be a permutation-invariant state with the decomposition
		\begin{align}
			\sigma_n = \bigoplus_{\lambda \in \Lambda_n} \omega_{\lambda}(\sigma_n) \otimes \pi_{\cV_\lambda}.
		\end{align}
		Let $\omega_{\lambda}(\sigma_n) = \sum_j c_{j,\lambda} |\psi_{j,\lambda}\rangle\!\langle\psi_{j,\lambda}|$ be the spectral decomposition. Define the projections $P_{j,\lambda} := |\psi_{j,\lambda}\rangle\!\langle\psi_{j,\lambda}| \otimes I_{\cV_\lambda}$, and the associated pinching map
		\begin{align}
			\mathcal{E}_n(X) := \sum_{\lambda \in \Lambda_n} \sum_j P_{j,\lambda} X P_{j,\lambda},
		\end{align}
		for any operator $X$.
		Then for any $r \in \RR$ and $\epsilon > 0$, there exist states
		$\rho_{j,\lambda}\in\sA_{n}$ (one for each index pair) and tests
		$T_{j,\lambda}$,
		such that:
		\begin{enumerate}[label=\textnormal{(\roman*)}]
			\item (Approximate maximality)
			\begin{align}
				\tr[P_{j,\lambda}\rho_{j,\lambda}] \geq \sup_{\rho_n\in\sA_n}\tr[P_{j,\lambda}\rho_{n}] - \epsilon, \quad \forall (j,\lambda); \label{eq: approx max}
			\end{align}
			\item (Blockwise Type-II bound)
			\begin{align}
				\tr[\sigma_{n}\, T_{j,\lambda}] & \le 2^{-nr}; \label{eq: feasible test tmp6}
			\end{align}
			\item (Blockwise Type-I bound)
			\begin{align}
				\tr\!\left[{\cal E}_{n}(\rho_{j,\lambda})(I-T_{j,\lambda})\right] & \le d_{n}^{\frac{1-\alpha}{\alpha}}\, 2^{\frac{1-\alpha}{\alpha}n\left(r-\frac{1}{n}D_{\Sand,\alpha}(\sA_n\|\sigma_n)\right)}. \label{eq: blockwise type I}
			\end{align}
		\end{enumerate}
	\end{boxlemma}

	\begin{proof}
		The $\epsilon$-approximate maximizer in item (i) always exists by the definition of the supremum.
		For each $(j,\lambda)$ define the threshold
		\begin{align}
			\label{eq: Rjlambda}R_{j,\lambda}:= \frac{nr+\log Q_{\Petz,\alpha}({\cal E}_{n}(\rho_{j,\lambda})\|\sigma_{n})}{n\alpha}.
		\end{align}
		Consider the Neyman--Pearson test between the pinched state ${\cal E}_{n}(\rho_{j,\lambda})$ and $\sigma_{n}$:
		\begin{align}
			T_{j,\lambda}:=\Bigl\{{\cal E}_{n}(\rho_{j,\lambda}) \geq 2^{nR_{j,\lambda}}\sigma_{n}\Bigr\}.
		\end{align}
		Recall that for any $V,W\in\cH_{+}$ and $\alpha\in(0,1)$, the following holds~\cite{audenaert2007discriminating}:
		\begin{align}
			\tr[V^{\alpha} W^{1-\alpha}] \geq \tr W\{W\le V\} + \tr V\{W > V\}.
		\end{align}
		Applying this with $V ={\cal E}_{n}(\rho_{j,\lambda})$ and
		$W = 2^{nR_{j,\lambda}}\sigma_{n}$, we obtain
		\begin{align}
			\tr\bigl[2^{nR_{j,\lambda}}\sigma_{n}\, T_{j,\lambda}\bigr] + \tr\bigl[{\cal E}_{n}(\rho_{j,\lambda})(I - T_{j,\lambda})\bigr] \le 2^{n(1-\alpha)R_{j,\lambda}}\, Q_{\Petz,\alpha}({\cal E}_{n}(\rho_{j,\lambda})\|\sigma_{n}).
		\end{align}
		Since both terms on the left-hand side are non-negative, each is
		individually bounded by the right-hand side:
		\begin{align}
			\tr[\sigma_{n}\, T_{j,\lambda}]                                  & \le 2^{-n\alpha R_{j,\lambda}}\, Q_{\Petz,\alpha}({\cal E}_{n}(\rho_{j,\lambda})\|\sigma_{n}), \label{eq: NP type II raw}  \\
			\tr\bigl[{\cal E}_{n}(\rho_{j,\lambda})(I - T_{j,\lambda})\bigr] & \le 2^{n(1-\alpha)R_{j,\lambda}}\, Q_{\Petz,\alpha}({\cal E}_{n}(\rho_{j,\lambda})\|\sigma_{n}). \label{eq: NP type I raw}
		\end{align}
		Substituting~\eqref{eq: Rjlambda} into~\eqref{eq: NP type II raw} yields
		item~(ii). Substituting into~\eqref{eq: NP type I raw} gives
		\begin{align}
			\tr\!\left[{\cal E}_{n}(\rho_{j,\lambda})(I-T_{j,\lambda})\right] & \le 2^{\frac{1-\alpha}{\alpha}n\left(r-\frac{1}{n}D_{\Petz,\alpha}({\cal E}_n(\rho_{j,\lambda})\|\sigma_n)\right)}\\
			& \le 2^{\frac{1-\alpha}{\alpha}n\left(r-\frac{1}{n}D_{\Sand,\alpha}({\cal E}_n(\rho_{j,\lambda})\|\sigma_n)\right)},
		\end{align}
		where the second inequality follows since 
		$D_{\Petz,\alpha}\geq D_{\Sand,\alpha}$ for $\alpha\in(0,1)$ (e.g.~\cite[Eq.~(7.5.45)]{khatri2024principlesquantumcommunicationtheory}).

		By~\cite[Lemma 3]{hayashi2016correlation}, we have the relation for the sandwiched \Renyi divergence between the original state and the pinched state:
		\begin{align}
			D_{\Sand,\alpha}(\rho_{j,\lambda}\|\sigma_{n}) \leq D_{\Sand,\alpha}({\cal E}_{n}(\rho_{j,\lambda})\|\sigma_{n})+\log d_{n} . \label{eq: pinching Sand bound}
		\end{align}

		Since $\rho_{j,\lambda}\in \sA_{n}$,
		the definition of the divergence between a set and a state
		(Definition~\ref{def: divergence between two sets}) gives
		$D_{\Sand,\alpha}(\rho_{j,\lambda}\|\sigma_{n}) \geq D_{\Sand,\alpha}(\sA_{n}
		\|\sigma_{n})$. Combining these two bounds establishes item~(iii).
	\end{proof}

	Using the blockwise tests from Lemma~\ref{lem: blockwise NP tests}, we now
	assemble them into a single universal test with uniform error control against all states in $\sA_{n}$.

	\begin{boxlemma}
		[Universal test via Schur--Weyl pinching.]\label{lem: universal test pinching}
		Let $\cH$ be a finite-dimensional Hilbert space, $n \in \NN$, and
		$\alpha \in (0,1)$. Let $\sA_{n} \subseteq \density(\cH^{\ox n})$
		be a set of quantum states and $\sigma_{n} \in \density(\cH^{\ox n})$ be a
		permutation-invariant state. 
		Then for any $r \in \RR$, there exists a test $0 \le M_{n} \le I$ such that:
		\begin{enumerate}[label=\textnormal{(\roman*)}]
			\item (Universal Type-II bound)
			\begin{align}
				\tr[\sigma_{n}\, M_{n}]  \le 2^{-nr}; \label{eq: universal test type II}
			\end{align}
			\item (Universal Type-I bound)
			\begin{align}
				\sup_{\rho_n \in \sA_n}\tr[\rho_{n}(I - M_{n})] \le d_{n}^{\frac{2}{\alpha}}\, 2^{\frac{1-\alpha}{\alpha}n\left(r - \frac{1}{n}D_{\Sand,\alpha}(\sA_{n}\|\sigma_{n})\right)}. \label{eq: universal test type I}
			\end{align}
		\end{enumerate}
	\end{boxlemma}

	\begin{proof}
		We use the same notation as in Lemma~\ref{lem: blockwise NP tests}.
		Define an operator $M_{n}$ by
		\begin{align}
			M_{n} := \sum_{j',\lambda'}m_{j',\lambda'}\, P_{j',\lambda'}, \quad \text{where}\quad m_{j',\lambda'}:= \max_{j,\lambda}\frac{\tr\!\left[T_{j,\lambda}P_{j',\lambda'}\right]}{\tr P_{j',\lambda'}}. \label{eq: Mn def}
		\end{align}
		We then show that this is the universal test satisfying items (i) and (ii). 
		Before doing so, we first show the relations that
		\begin{align}
			T_{j,\lambda} \leq M_{n} \leq \sum_{j,\lambda}T_{j,\lambda}, \qquad \forall (j,\lambda). \label{eq: Mn bounds}
		\end{align}

		First note that $\sigma_{n}$ is diagonal in $\{P_{j,\lambda}
		\}_{j,\lambda}$ by construction,
		and ${\cal E}_{n}(\rho_{j,\lambda})$ is also diagonal in the same
		decomposition because ${\cal E}_{n}$ is the corresponding pinching map. Hence the Neyman--Pearson test $T_{j,\lambda}=\bigl\{{\cal E}_{n}(\rho_{j,\lambda}) \geq 2^{nR_{j,\lambda}}
		\sigma_{n}\bigr\}$ is also diagonal in the
		same decomposition. Therefore, for each fixed $(j,\lambda)$, there exist
		coefficients $t^{(j,\lambda)}_{j',\lambda'}\in\{0,1\}$ such that
		\begin{align}
			T_{j,\lambda}=\sum_{j',\lambda'}t^{(j,\lambda)}_{j',\lambda'}\,P_{j',\lambda'}.
		\end{align}
		In particular,
		\begin{align}
			\frac{\tr[T_{j,\lambda}P_{j',\lambda'}]}{\tr P_{j',\lambda'}}=t^{(j,\lambda)}_{j',\lambda'}\in\{0,1\}.
		\end{align}

		This implies that
		\begin{align}
			m_{j',\lambda'}=\max_{j,\lambda}t^{(j,\lambda)}_{j',\lambda'}\in\{0,1\}.\label{eq: Mn leq sum P}
		\end{align}
		Fix any $(j,\lambda)$. Then for every $(j',\lambda')$ we have $t^{(j,\lambda)}
		_{j',\lambda'}\le m_{j',\lambda'}$, and thus
		\begin{align}
			M_{n}-T_{j,\lambda}= \sum_{j',\lambda'}\bigl(m_{j',\lambda'}-t^{(j,\lambda)}_{j',\lambda'}\bigr)P_{j',\lambda'}\ge 0.
		\end{align}
		This proves the lower bound in~\eqref{eq: Mn bounds}.

		From ~\eqref{eq: Mn leq sum P}, we have
		$m_{j',\lambda'}\le \sum_{j,\lambda}t^{(j,\lambda)}_{j',\lambda'}$ for each
		$(j',\lambda')$. Therefore,
		\begin{align}
			M_{n} = \sum_{j',\lambda'}m_{j',\lambda'}\, P_{j',\lambda'}\le \sum_{j',\lambda'}\sum_{j,\lambda}t^{(j,\lambda)}_{j',\lambda'}\, P_{j',\lambda'}= \sum_{j,\lambda}T_{j,\lambda}.
		\end{align}
		This proves the upper bound in~\eqref{eq: Mn bounds}.

		We now check that $M_{n}$ satisfies items (i) and (ii). For any $r \in \RR$, we define 
		\begin{align} 
			\tilde{r} := r + \frac{1}{n}\log d_{n},
		\end{align} 
		and apply
		Lemma~\ref{lem: blockwise NP tests} with this $\tilde{r}$ and an
		arbitrary $\epsilon > 0$.
		Since $\sigma_{n} \ge 0$, this gives
		\begin{align}
			\tr[\sigma_{n} M_{n}] & \le \sum_{j,\lambda}\tr[\sigma_{n} T_{j,\lambda}] \le \sum_{j,\lambda}2^{-n\tilde{r}}= d_{n}\,2^{-n\tilde{r}}= 2^{-nr}. \label{eq: assembled type II}
		\end{align}
		where the second inequality follows from~\eqref{eq: feasible test tmp6}. This confirms item (i). 
		
		Moreover, we have for every $(j,\lambda)$,
		\begin{align}
			\tr[\cE_{n}(\rho_{j,\lambda})(I-M_{n})] & \le \tr\!\left[\cE_{n}(\rho_{j,\lambda})(I-T_{j,\lambda})\right]\\
			& \le d_{n}^{\frac{1-\alpha}{\alpha}}\, 2^{\frac{1-\alpha}{\alpha}n\left(\tilde{r}-\frac{1}{n}D_{\Sand,\alpha}(\sA_n\|\sigma_n)\right)}\\
			& = d_{n}^{\frac{2(1-\alpha)}{\alpha}}\, 2^{\frac{1-\alpha}{\alpha}n\left(r-\frac{1}{n}D_{\Sand,\alpha}(\sA_n\|\sigma_n)\right)},\label{eq: assembled blockwise type I}
		\end{align}
		where the first line uses the lower bound in~\eqref{eq: Mn bounds} and the second line applies~\eqref{eq: blockwise type I} with $\tilde{r}$ in place of $r$ and the last line uses the definition of $\tilde{r}$.

		We now bound the Type-I error uniformly over $\sA_{n}$. For any $\rho_{n}\in\sA
		_{n}$:
		\begin{align}
			\tr[\rho_{n}(I-M_{n})] & \leq d_{n}\,\tr[{\cal E}_{n}(\rho_{n})(I-M_{n})] \\
			& = d_{n} \sum_{j,\lambda}\tr[\rho_{n} P_{j,\lambda}]\,(1 - m_{j,\lambda})                                                                                                  \\
			& \leq d_{n}\sum_{j,\lambda}\bigl(\tr[\rho_{j,\lambda}P_{j,\lambda}] + \epsilon\bigr)\,(1 - m_{j,\lambda})                                                                  \\
			& \leq d_{n}\sum_{j,\lambda}\tr[\rho_{j,\lambda}P_{j,\lambda}]\,(1 - m_{j,\lambda}) + d_{n}^{2} \epsilon                                                                    \\
			& \leq d_{n}\sum_{j,\lambda}\tr\!\left[{\cal E}_{n}(\rho_{j,\lambda})(I-M_{n})\right] + d_{n}^{2}\,\epsilon                                                                 \\
			& \leq d_{n}\sum_{j,\lambda}d_{n}^{\frac{2(1-\alpha)}{\alpha}}\, 2^{\frac{1-\alpha}{\alpha}n\left(r-\frac{1}{n}D_{\Sand,\alpha}(\sA_n\|\sigma_n)\right)}+ d_{n}^{2}\,\epsilon  \\
			& \leq d_{n}^{\frac{2}{\alpha}}\, 2^{\frac{1-\alpha}{\alpha}n\left(r-\frac{1}{n}D_{\Sand,\alpha}(\sA_n\|\sigma_n)\right)}+ d_{n}^{2}\,\epsilon. \label{eq: uniform type I eps}
		\end{align}
		The first line is the pinching inequality that $\rho_n \leq d_n \cE_n(\rho_n)$ (see e.g.~\cite[Lemma 3.10]{hayashi2017quantum}). For the second line, we
		expand $\tr[{\cal E}_{n}(\rho_{n})(I-M_{n})]$ using the block decomposition and
		the orthogonality of the projections $P_{j,\lambda}$ gives
		\begin{align}
			\tr[{\cal E}_{n}(\rho_{n})(I-M_{n})] & = \sum_{j,\lambda}\tr[\rho_{n} P_{j,\lambda}]\,(1 - m_{j,\lambda}).
		\end{align}
		The third line uses the approximate maximality of $\rho_{j,\lambda}$: by~\eqref{eq: approx max},
		$\tr[\rho_{n} P_{j,\lambda}] \leq \tr[\rho_{j,\lambda}P_{j,\lambda}] + \epsilon$
		for every $(j,\lambda)$. The fourth line pulls out the constant $\epsilon$ outside
		of summation. For the fifth line, we use the same block-diagonal expansion applied
		to $\rho_{j,\lambda}$:
		\begin{align}
			\tr[{\cal E}_{n}(\rho_{j,\lambda})(I-M_{n})] = \sum_{j',\lambda'}\tr[\rho_{j,\lambda}P_{j',\lambda'}](1 - m_{j',\lambda'}) \geq \tr[\rho_{j,\lambda}P_{j,\lambda}](1 - m_{j,\lambda}),
		\end{align}
		where the inequality holds because every term
		$\tr[\rho_{j,\lambda}P_{j',\lambda'}](1 - m_{j',\lambda'})$ in the sum is
		non-negative. The sixth line applies~\eqref{eq: assembled blockwise type I}, and the last
		uses $\sum_{j,\lambda}1 \leq d_{n}$.

		As~\eqref{eq: uniform type I eps} holds for any $\epsilon > 0$ and $\rho_n \in \sA_n$, we conclude
		\begin{align}
			\sup_{\rho_n \in \sA_n}\tr[\rho_{n}(I-M_{n})] \leq d_{n}^{\frac{2}{\alpha}}\,2^{\frac{1-\alpha}{\alpha}n\left(r-\frac{1}{n}D_{\Sand,\alpha}(\sA_n\|\sigma_n)\right)}.
		\end{align}
		This confirms item (ii) and completes the proof.
	\end{proof}

	We are now in a position to state the quantum Hoeffding lower bound. We provide two versions, one in terms of the sandwiched \Renyi divergence in Proposition~\ref{prop: quantum Hoeffding lower bound} and the other in terms of the reverse sandwiched \Renyi divergence in Theorem~\ref{thm: quantum Hoeffding lower bound}. 

	\begin{boxproposition}
		[Quantum Hoeffding lower bound via $D_{\Sand,\alpha}$.]\label{prop: quantum Hoeffding lower bound} Let $\cH$ be a finite-dimensional
		Hilbert space. Let $\sA = \{\sA_{n}\}_{n\in \NN}$ and
		$\sB = \{\sB_{n}\}_{n\in \NN}$ be sequences of sets of quantum states such that $\sA_{n}, \sB_{n} \subseteq \density(\cH^{\ox n})$ for every $n \in \NN$. Suppose that
		\begin{itemize}
			\item $\sA$ satisfies assumption (C5$'$) $\perm$-invariance;

			\item $\sB$ satisfies assumptions (C1) convexity and (C5) $\perm$-closedness.
		\end{itemize}
		Then, for every $r>0$, it holds that
		\begin{align}
			\liminf_{n\to \infty}-\frac{1}{n}\log \alpha_{n, r}(\sA_{n}\| \sB_{n}) & \geq \sup_{\alpha \in (0,1)}\frac{\alpha-1}{\alpha}\bigg(r - \underline{D}_{\Sand,\alpha}^{\infty}(\sA\|\sB)\bigg).  \label{NK3}
		\end{align}
	\end{boxproposition}

	\begin{proof}

		By Lemma~\ref{lem: reduction perm inv} and the required assumptions, we have
		\begin{align}\label{eq: lower bound proof tmp1}
			\alpha_{n,r}(\sA_{n}\|\sB_{n}) = \sup_{\sigma_n\in\cT_{\perm}(\sB_{n})}\alpha_{n,r}(\sA_{n}\|\sigma_{n}).
		\end{align}
		Fix any $\sigma_{n}\in\cT_{\perm}(\sB_{n})$ and $\alpha \in (0,1)$.
		Lemma~\ref{lem: universal test pinching} implies that
		\begin{align}
			\alpha_{n,r}(\sA_{n}\|\sigma_{n}) \leq d_{n}^{\frac{2}{\alpha}}\, 2^{\frac{1-\alpha}{\alpha}n\left(r-\frac{1}{n}D_{\Sand,\alpha}(\sA_n\|\sigma_n)\right)}.
		\end{align}
		By the assumptions for $\sB_n$, we have the inclusion that $\cT_{\group}(\sB_n) \subseteq \sB_n$. So $\sigma_n \in \sB_n$ and therefore $D_{\Sand,\alpha}(\sA_{n}\|\sigma_{n}) \ge D_{\Sand,\alpha}(\sA_{n}\|\sB_{n})$ by definition, so
		\begin{align}
			\alpha_{n,r}(\sA_{n}\|\sigma_{n}) \leq d_{n}^{\frac{2}{\alpha}}\, 2^{\frac{1-\alpha}{\alpha}n\left(r-\frac{1}{n}D_{\Sand,\alpha}(\sA_n\|\sB_n)\right)}.
		\end{align}
		Since the right-hand side does not depend on the specific choice of $\sigma_{n} \in \cT_{\perm}(\sB_{n})$, taking the supremum and applying~\eqref{eq: lower bound proof tmp1} yields
		\begin{align}
			\alpha_{n,r}(\sA_{n}\|\sB_{n}) \leq d_{n}^{\frac{2}{\alpha}}\, 2^{\frac{1-\alpha}{\alpha}n\left(r-\frac{1}{n}D_{\Sand,\alpha}(\sA_n\|\sB_n)\right)}.
		\end{align}
		Taking $-\frac{1}{n}\log$ of both sides and using $\frac{\log d_{n}}{n}\to 0$
		as $n \to \infty$, we obtain
		\begin{align}
			\liminf_{n\to \infty}-\frac{1}{n}\log \alpha_{n, r}(\sA_{n}\| \sB_{n}) \geq \frac{\alpha-1}{\alpha}\Big(r - \underline{D}_{\Sand,\alpha}^{\infty}(\sA\|\sB)\Big).
		\end{align}
		Since this holds for every $\alpha \in (0,1)$, taking the supremum over $\alpha$
		concludes the proof.
	\end{proof}

		It is interesting that this lower bound closely parallels the strong converse bound in~\cite[Theorem 32]{fang2025error},~\footnote{The presented bound combines~\cite[Theorem 32]{fang2025error} and~\cite[Proof of Lemma 23]{fang2025error}.}
	\begin{align}\label{eq: strong converse}
		\liminf_{n\to \infty}-\frac{1}{n}\log \left(1-\alpha_{n,r}(\sA_{n}\|\sB_{n})\right) \geq \sup_{\alpha > 1}\frac{\alpha-1}{\alpha}\bigg(r - \overline{D}_{\Sand,\alpha}^{\infty}(\sA\|\sB)\bigg).
	\end{align}
	Both~\eqref{NK3} and~\eqref{eq: strong converse} are expressed in terms of sandwiched \Renyi divergences, but over complementary parameter ranges. 

	\begin{boxtheorem}
		[Quantum Hoeffding lower bound via $D_{\RSand,\alpha}$.]\label{thm: quantum Hoeffding lower bound} Let $\cH$ be a finite-dimensional
		Hilbert space. Let $\sA = \{\sA_{n}\}_{n\in \NN}$ and
		$\sB = \{\sB_{n}\}_{n\in \NN}$ be sequences of sets of quantum states such that $\sA_{n}, \sB_{n} \subseteq \density(\cH^{\ox n})$ for every $n \in \NN$. Suppose that
		\begin{itemize}
			\item $\sA$ satisfies assumptions (C1) convexity and (C5) $\perm$-closedness;

			\item $\sB$ satisfies assumption (C5$'$) $\perm$-invariance.
		\end{itemize}
		Then, for every $r>0$, it holds that
		\begin{align}
			\liminf_{n\to \infty}-\frac{1}{n}\log \alpha_{n, r}(\sA_{n}\| \sB_{n}) & \geq \sup_{\alpha \in (0,1)}\frac{\alpha-1}{\alpha}\bigg(r - \underline{D}_{\RSand,\alpha}^{\infty}(\sA\|\sB)\bigg). \label{NK4}
		\end{align}
	\end{boxtheorem}
	\begin{proof}
		The proof follows similar steps to those of Proposition~\ref{prop: quantum Hoeffding lower bound}, but with the roles of $\sA$ and $\sB$ exchanged.
		Fix $\alpha \in (0,1)$ and $r > 0$.
		Applying Lemma~\ref{lem: reduction perm inv} with $\sA$ and $\sB$ swapped (here $\sA$ satisfies~(C1) and~(C5), and $\sB$ satisfies~(C5$'$)) gives
		\begin{align}
			\alpha_{n,r}(\sA_{n}\|\sB_{n}) = \sup_{\rho_n\in\cT_{\perm}(\sA_{n})}\alpha_{n,r}(\rho_{n}\|\sB_{n}).
		\end{align}
		Define the auxiliary rate
		\begin{align}
			r' := -\frac{\alpha}{1-\alpha}\, r + \frac{1}{n}D_{\Sand,\alpha}(\sB_{n}\|\sA_{n}) - \frac{2}{\alpha}\,\frac{\log d_{n}}{n},
		\end{align}
		which satisfies
		\begin{align}
			2^{-nr'}  & = d_{n}^{\frac{2}{\alpha}}\, 2^{-n\left(-\frac{\alpha r}{1-\alpha} + \frac{1}{n}D_{\Sand,\alpha}(\sB_n\|\sA_n)\right)},\label{eq: 2 nr'}       \\
			2^{-nr} & = d_{n}^{\frac{2}{\alpha}}\, 2^{\frac{1-\alpha}{\alpha}n\left(r' - \frac{1}{n}D_{\Sand,\alpha}(\sB_n\|\sA_n)\right)}.\label{eq: r' r relation}
		\end{align}
		Fix any $\rho_{n}\in\cT_{\perm}(\sA_{n})$. Applying Lemma~\ref{lem: universal test pinching} with the set $\sB_{n}$, the permutation-invariant state $\rho_{n}$, and rate $r'$ yields a test $0 \le M_{n} \le I$ satisfying
		\begin{align}
			\tr[\rho_{n}\, M_{n}]                                  & \le 2^{-nr'}, \label{eq: coro type II}\\
			\sup_{\sigma_n \in \sB_n}\tr[\sigma_{n}(I - M_{n})] & \le d_{n}^{\frac{2}{\alpha}}\, 2^{\frac{1-\alpha}{\alpha}n\left(r' - \frac{1}{n}D_{\Sand,\alpha}(\sB_n\|\rho_n)\right)}. \label{eq: coro type I}
		\end{align}
		Since $\sA_n$ is convex and permutation-closed, we have $\rho_{n} \in \sA_{n}$, and hence $D_{\Sand,\alpha}(\sB_{n}\|\rho_{n}) \ge D_{\Sand,\alpha}(\sB_{n}\|\sA_{n})$. Substituting this together with~\eqref{eq: 2 nr'} and~\eqref{eq: r' r relation} into the above bounds gives
		\begin{align}
			\tr[\rho_{n}\, M_{n}]    & \leq d_{n}^{\frac{2}{\alpha}}\, 2^{-n\left(-\frac{\alpha r}{1-\alpha} + \frac{1}{n}D_{\Sand,\alpha}(\sB_n\|\sA_n)\right)},\label{eq: coro type II simplified} \\
			\sup_{\sigma_n \in \sB_n}\tr[\sigma_{n}(I - M_{n})] & \leq 2^{-nr}. \label{eq: coro type I simplified}
		\end{align}

		The bound~\eqref{eq: coro type I simplified} shows that $I - M_{n}$ is a feasible test for the optimization defining $\alpha_{n,r}(\rho_n\|\sB_{n})$, with corresponding Type-I error $\tr[\rho_{n}\, M_{n}]$ controlled by~\eqref{eq: coro type II simplified}. Therefore,
		\begin{align}
			\label{eq: coro alpha bound}\alpha_{n,r}(\rho_{n}\|\sB_{n}) \leq d_{n}^{\frac{2}{\alpha}}\, 2^{-n\left(-\frac{\alpha r}{1-\alpha} + \frac{1}{n}D_{\Sand,\alpha}(\sB_n\|\sA_n)\right)}.
		\end{align}
		Since the right-hand side is independent of $\rho_{n} \in \cT_{\perm}(\sA_{n})$, taking the supremum yields
		\begin{align}
			\alpha_{n,r}(\sA_{n}\|\sB_{n}) \leq d_{n}^{\frac{2}{\alpha}}\, 2^{-n\left(-\frac{\alpha r}{1-\alpha} + \frac{1}{n}D_{\Sand,\alpha}(\sB_n\|\sA_n)\right)}.
		\end{align}
		Taking $-\frac{1}{n}\log$ of both sides and using $\frac{\log d_{n}}{n}\to 0$
		as $n \to \infty$, we obtain
		\begin{align}
			\label{eq: coro Sand form}
			\liminf_{n\to \infty}-\frac{1}{n}\log \alpha_{n,r}(\sA_{n}\|\sB_{n}) \geq \frac{\alpha r}{\alpha - 1}+ \underline{D}_{\Sand,\alpha}^{\infty}(\sB\|\sA).
		\end{align}
		To convert this into the reverse sandwiched form, we use the definition~\eqref{eq: sandwiched-R}, which gives
		\begin{align}
			\underline{D}_{\Sand,\alpha}^{\infty}(\sB\|\sA) = \frac{\alpha}{1-\alpha}\,\underline{D}_{\RSand,1-\alpha}^{\infty}(\sA\|\sB).
		\end{align}
		Substituting into~\eqref{eq: coro Sand form} yields
		\begin{align}
			\liminf_{n\to \infty}-\frac{1}{n}\log \alpha_{n,r}(\sA_{n}\|\sB_{n}) \geq \frac{\alpha}{\alpha-1}\Big(r - \underline{D}_{\RSand,1-\alpha}^{\infty}(\sA\|\sB)\Big).
		\end{align}
		Since this holds for every $\alpha \in (0,1)$, replacing $\alpha$ by $1-\alpha$ and taking the supremum over $\alpha$ establishes the lower bound.
	\end{proof}

	The above estimate, or more explicitly its form in~\eqref{eq: coro Sand form}, generalizes~\cite[Theorem 2]{Hayashi2025entanglement}, which treats the special case of testing i.i.d.\ copies of a state against the set of all separable states.
	This lower bound is not expected to be tight in general: for simple i.i.d.\ null and alternative hypotheses, it fails to recover the Hoeffding exponent governed by the Petz \Renyi divergence. Crucially for the present work, however, the reverse sandwiched form in~\eqref{NK4} is tight enough to yield a matching bound in Section~\ref{sec: matching bounds} for a class of composite hypothesis testing problems, including the one that underlies the operational interpretation of the reverse sandwiched \Renyi divergence.

	\subsection{Quantum Hoeffding upper bound}\label{sec: quantum Hoeffding upper bound}

	We next establish the quantum Hoeffding upper bound. As in the previous subsection, we first develop the technical ingredients and then state and prove Theorem~\ref{thm: quantum Hoeffding upper bound}. The strategy is to reduce the quantum problem to a classical large-deviation analysis and then lift it to composite sets. Specifically, the Nussbaum--Szko\l{}a distributions recast the testing problem in terms of classical log-likelihood ratios, to which a tailored G\"{a}rtner--Ellis theorem (Lemma~\ref{lem: Gartner-Ellis}) supplies the basic large-deviation estimate. This then leads to Proposition~\ref{lem: Hoeffding bound sets}, which controls the optimal Type-I exponent for two state sequences in terms of their regularized Petz \Renyi divergence. Finally, Lemma~\ref{lem: averaged state pair} constructs a universal tight pair by averaging $\alpha$-dependent near-minimizers over a refining grid, so that a single sequence simultaneously attains the regularized divergence at every $\alpha \in (0,1)$ of continuity, lifting the bound from individual pairs to composite sets.

	Let the spectral decompositions of $\rho$ and $\sigma$ be given by
	\begin{align}
		\rho = \sum_{i=1}^{d} \lambda_{i} \ketbra{u_i}{u_i}\quad \text{and}\quad \sigma = \sum_{j=1}^{d} \mu_{j} \ketbra{v_j}{v_j},
	\end{align}
	where $\ket{u_i}$ and $\ket{v_j}$ are two orthonormal bases and $\lambda_{i}$ and
	$\mu_{j}$ are the corresponding eigenvalues, respectively. Then the Nussbaum-Szko\l{}a
	distributions of $\rho,\sigma$ are defined by~\cite{Nussbaum2009}
	\begin{align}
		(P_{\rho,\sigma})(i,j) = \lambda_{i}|\<u_{i}|v_{j}\>|^{2} \quad \text{and}\quad (Q_{\rho,\sigma})(i,j) = \mu_{j} |\<u_{i}|v_{j}\>|^{2},
	\end{align}
	where $i,j \in \{1,\ldots, d\}$.

	Given a sequence of random variables
	$\{X_{n}\}_{n\in \NN}$, the asymptotic cumulant generating function is defined
	as
	\begin{align}
		\Lambda_{X}(t):= \lim_{n\to \infty}\frac{1}{n}\log \mathbb{E}\left[\exp({nt X_n})\right],
	\end{align}
	provided that the limit exists.

	We will use the G\"{a}rtner--Ellis theorem~\cite[Theorem~2.3.6]{Dembo2010} in the following slightly tailored form, which suffices for our later arguments.

	\begin{lemma}
		[{\cite[Lemma 29]{fang2025error}}.]\label{lem: Gartner-Ellis} Assume that the
		asymptotic cumulant generating function $t \mapsto \Lambda_{X}(t)$ exists
		for all $t\in\RR$, and that $\Lambda_{X}$ is strictly convex and $C^{1}$-continuous on $(a,b)$. Fix an open
		interval $(a,b)\subseteq \RR$. Then, for any
		$x \in \bigl(\Lambda_{X}'(a),\, \Lambda_{X}'(b)\bigr)$,
		\begin{align}
			\limsup_{n\to \infty}-\frac{1}{n}\log \Pr\{X_{n} > x\} \leq \sup_{t\in (a,b)}\{tx - \Lambda_{X}(t)\}.
		\end{align}
	\end{lemma}

	The following lemma extends~\cite[Lemma 30]{fang2025error} to general sequences
	of states.

	\begin{boxlemma}
		\label{LH1} Let $\cH$ be a finite-dimensional Hilbert space and let $\rho^{(n)}
		, \sigma^{(n)}\in \density(\cH^{\ox n})$ be sequences of quantum states.
		Define
		\begin{align}
			\phi(s):= \lim_{n\to \infty}\frac{1}{n}\log \tr (\rho^{(n)})^{1-s}(\sigma^{(n)})^{s},
		\end{align}
		and assume that $\phi(s)$ exists and is finite for each $s\in(0,1)$. Assume further
		that $\phi(s)$ is strictly convex and $C^{1}$-continuous on $(0,1)$.
		If $R \in (-\phi'(1), -\phi'(0))$, then for any $0 \leq T_{n} \leq I$,
		\begin{align}
			\limsup_{n\to \infty}-\frac{1}{n}\log \left[\tr e^{-nR}\rho^{(n)}(I-T_{n}) +\tr\sigma^{(n)}T_{n} \right] \leq \max_{s\in (0,1)}(1-s)R -\phi(s).
		\end{align}
	\end{boxlemma}

	\begin{proof}
		Let ${P^{(n)}}$ and ${Q^{(n)}}$ be the Nussbaum-Szko\l{}a distributions of
		$\rho^{(n)}$ and $\sigma^{(n)}$. Let
		\begin{align}
			S_{n} = \left\{e^{-nR }{P^{(n)}}>{Q^{(n)}}\right\},
		\end{align}
		be a likelihood ratio test. Consider the random variable
		\begin{align}
			X_{n}(x) := \frac{1}{n}\left(\log{Q^{(n)}}(x) - \log{P^{(n)}}(x)\right),
		\end{align}
		where $x$ is drawn from the distribution ${P^{(n)}}$. Then we have the asymptotic
		cumulant generating function of the random variable $X_{n}$ as,
		\begin{align}
			\lim_{n\to \infty} & \frac{1}{n}\log \sum_{x}{P^{(n)}}(x) \exp \left(s n X_{n}(x) \right)    \\
			                   & = \lim_{n\to \infty}\frac{1}{n}\log \tr{(Q^{(n)})}^{s}{(P^{(n)})}^{1-s}       \\
			                   & = \lim_{n\to \infty}\frac{1}{n}\log \tr (\sigma^{(n)})^{s} (\rho^{(n)})^{1-s} \\
			                   & =\phi(s),
		\end{align}
		where the second equality is a standard property of the Nussbaum-Szko\l{}a distributions
		(cf.~\cite[Proposition 1]{audenaert2008asymptotic}). 
		Applying the G\"{a}rtner--Ellis theorem in Lemma~\ref{lem: Gartner-Ellis}
		for the random variable $X_{n}$, interval $(0,1)$ and $x=-R$, we have
		\begin{align}
			\limsup_{n\to \infty}- \frac{1}{n}\log \Pr\{X_{n} \geq -R\} \leq \sup_{s\in (0,1)}-s R -\phi(s). \label{eq: NS relation tmp1}
		\end{align}
		Similarly, consider the random variable
		\begin{align}
			Y_{n}(x):= \frac{1}{n}\left(\log{P^{(n)}}(x) - \log{Q^{(n)}}(x)\right),
		\end{align}
		where $x$ is drawn from the distribution ${Q^{(n)}}$. Then we have the
		asymptotic cumulant generating function of the random variable $Y_{n}$ as,
		\begin{align}
			\lim_{n\to \infty} & \frac{1}{n}\log \sum_{x}{Q^{(n)}}(x) \exp \left(t n Y_{n}(x)\right)    \\
			& = \lim_{n\to \infty}\frac{1}{n}\log \tr{(Q^{(n)})}^{1-t}{(P^{(n)})}^{t}      \\
			& = \lim_{n\to \infty}\frac{1}{n}\log \tr (\sigma^{(n)})^{1-t}(\rho^{(n)})^{t} \\
			& =\phi(1-t).
		\end{align}
		Applying again the G\"{a}rtner--Ellis theorem in Lemma~\ref{lem: Gartner-Ellis}
		for the random variable $Y_{n}$, interval $(0,1)$ and $x=R$, we have
		\begin{align}
			\limsup_{n\to \infty}-\frac{1}{n}\log \Pr\{Y_{n} > R\} & \leq \sup_{t \in (0,1)}t R -\phi(1-t).\label{eq: NS relation tmp2}
		\end{align}
		By direct calculation, we have the relations
		\begin{align}
			\limsup_{n\to \infty}- \frac{1}{n}\log \Pr\{X_{n} \geq -R\} & = \limsup_{n\to \infty}-\frac{1}{n}\log \tr P^{(n)}(I-S_{n}),\label{eq: NS relation tmp3} \\
			\limsup_{n\to \infty}- \frac{1}{n}\log \Pr\{Y_{n} > R\}     & = \limsup_{n\to \infty}-\frac{1}{n}\log \tr Q^{(n)}S_{n}. \label{eq: NS relation tmp4}
		\end{align}
		Moreover, the Nussbaum-Szko\l{}a theorem (cf.~\cite[Lemma 3.4]{hayashi2017quantum})
		implies that for any test $T_{n}$,
		\begin{align}
			\label{eq: NS relation tmp5}\tr e^{-nR}\rho^{(n)}(I-T_{n}) & + \tr \sigma^{(n)}T_{n} \geq \frac{1}{2}\left(\tr e^{-nR }{P^{(n)}}(I- S_{n}) + \tr{Q^{(n)}}S_{n}\right).
		\end{align}
		Combining Eqs.~\eqref{eq: NS relation tmp1},~\eqref{eq: NS relation tmp2},~\eqref{eq: NS relation tmp3},~\eqref{eq: NS relation tmp4}
		and~\eqref{eq: NS relation tmp5}, we have
		\begin{align}
			\limsup_{n\to \infty} & -\frac{1}{n}\log \left[\tr e^{-nR}\rho^{(n)}(I-T_{n}) +\tr\sigma^{(n)}T_{n} \right]                          \\
			& \leq \limsup_{n\to \infty}- \frac{1}{n}\log \frac{1}{2}\left[\tr e^{-nR}P^{(n)}(I-S_{n}) +\tr{Q^{(n)}}S_{n}\right] \\
			& \leq \min \left\{ R+ \sup_{s\in (0,1)}-s R -\phi(s), \sup_{t \in (0,1)}t R -\phi(1-t)\right\}                      \\
			& =\sup_{s\in (0,1)}(1-s)R -\phi(s),\label{eq: lemma tmp1}
		\end{align}
		where the last equality follows by replacing $t$ with $1-s$. Finally, as the
		objective function in~\eqref{eq: lemma tmp1} is concave and its first
		derivative is given by $-R - \phi'(s)$, there is a unique critical point that
		achieves the maximum if $R \in (- \phi'(1), -\phi'(0))$. Therefore, the supremum
		is attained.
	\end{proof}

	\begin{boxlemma}
		\label{lem: phi continuity} Let $\cH$ be a finite-dimensional Hilbert space and
		let $\rho^{(n)}, \sigma^{(n)}\in \density(\cH^{\ox n})$ be sequences of quantum
		states. Define
		\begin{align}
			\phi_{n}(s):= \frac{1}{n}\log \tr (\rho^{(n)})^{1-s}(\sigma^{(n)})^{s}, \quad \text{and} \quad  \phi(s):= \lim_{n\to \infty}\phi_{n}(s),
		\end{align}
		and assume that $\phi(s)$ exists and is finite for each $s\in(0,1)$. Assume further
		that $\phi(s)$ is strictly convex and $C^{1}$-continuous on $(0,1)$. Consider two optimization problems:
		\begin{align}
			h_n(r) & := \max_{s\in (0,1)}\frac{-\phi_{n}(s)-s r}{1-s}, \quad \text{and} \quad 
			h(r)   := \max_{s\in (0,1)}\frac{-\phi(s)-s r}{1-s}.
		\end{align}
		Then
		for any
		$0< r < \lim_{s\to 0^+}-\phi'(s)$, we have that
		\begin{itemize}
			\item $h_n$ and $h$ are optimally achieved at unique critical points $s_{n,r}, s_r \in (0,1)$, respectively;
			\item $s_{n,r}$ and $s_r$ are continuous in $r$, and $s_{n,r} \to s_r$ as $n \to \infty$.
		\end{itemize} 
	\end{boxlemma}
	\begin{proof}
		We first show that the derivatives converge:
		\begin{align}
			\lim_{n\to \infty}\phi_{n}'(s)=\phi'(s) \label{BNA}
		\end{align}
		for $s\in (0,1)$. To see this, note that the convexity of $\phi_{n}$ gives, for
		any $\epsilon>0$,
		\begin{align}
			\frac{\phi_{n}(s)-\phi_{n}(s-\epsilon)}{\epsilon}\le \phi_{n}'(s) \le \frac{\phi_{n}(s+\epsilon)-\phi_{n}(s)}{\epsilon}.
		\end{align}
		Taking $\liminf$ and $\limsup$ as $n\to\infty$ (using the pointwise convergence
		$\phi_{n} \to \phi$), we obtain
		\begin{align}
			\frac{\phi(s)-\phi(s-\epsilon)}{\epsilon}\le \liminf_{n\to \infty}\phi_{n}'(s) \le \limsup_{n\to \infty}\phi_{n}'(s) \le \frac{\phi(s+\epsilon)-\phi(s)}{\epsilon}.
		\end{align}
		Sending $\epsilon\to 0$ and using the differentiability of $\phi$ yields~\eqref{BNA}.

		Define the objective function $f(s
		) := \frac{-\phi(s)-sr}{1-s}$ for $s\in(0,1)$. Its derivative is
		\begin{align}
			f'(s) = \frac{(s-1)\phi'(s) - \phi(s) - r}{(1-s)^{2}}.
		\end{align}
		Define the numerator as $g(s) := (s-1)\phi'(s) - \phi(s) - r$. 
		
		We claim that
		$g$ is strictly decreasing on $(0,1)$. Since $\phi$ is differentiable and
		strictly convex, it lies strictly above
		every tangent line~\cite[Eq. (3.3)]{boyd2004convex}: for any $a \neq b$ in $(0,1)$,
		\begin{align}
			\phi(a) > \phi(b) + \phi'(b)(a - b).\label{eq: strict convex tangent}
		\end{align}
		Now fix $0 < s_{1} < s_{2} < 1$. Applying~\eqref{eq: strict convex tangent}
		with $(a,b) = (s_{1}, s_{2})$ and $(a,b) = (s_{2}, s_{1})$ respectively:
		\begin{align}
			\phi(s_{1}) - \phi(s_{2}) &> \phi'(s_{2})(s_{1} - s_{2}), \label{eq: tangent 1}\\
			\phi(s_{2}) - \phi(s_{1}) &> \phi'(s_{1})(s_{2} - s_{1}). \label{eq: tangent 2}
		\end{align}
		Computing $g(s_{1}) - g(s_{2})$ directly:
		\begin{align}
			g(s_{1}) - g(s_{2}) &= (s_{1}-1)\phi'(s_{1}) - \phi(s_{1}) - (s_{2}-1)\phi'(s_{2}) + \phi(s_{2}).
		\end{align}
		Substituting the bound~\eqref{eq: tangent 2}, i.e.,
		$\phi(s_{2}) - \phi(s_{1}) > \phi'(s_{1})(s_{2} - s_{1})$, to replace
		$-\phi(s_{1}) + \phi(s_{2})$:
		\begin{align}
			g(s_{1}) - g(s_{2}) &> (s_{1}-1)\phi'(s_{1}) + \phi'(s_{1})(s_{2} - s_{1}) - (s_{2}-1)\phi'(s_{2}) \notag\\
			&= (s_{2}-1)\big(\phi'(s_{1}) - \phi'(s_{2})\big).
		\end{align}
		From~\eqref{eq: tangent 1} and~\eqref{eq: tangent 2}, adding the two
		inequalities gives $0 > (\phi'(s_{1}) - \phi'(s_{2}))(s_{2} - s_{1})$, so
		$\phi'(s_{1}) < \phi'(s_{2})$ (i.e., $\phi'$ is strictly increasing). Since
		$s_{2} - 1 < 0$ and $\phi'(s_{1}) - \phi'(s_{2}) < 0$, we obtain
		$g(s_{1}) - g(s_{2}) > 0$. Thus $g$ is strictly decreasing. 
		
		Since $\phi$ is $C^{1}$-continuous
		and strictly convex with $\phi(0) = \phi(1) = 0$, we have the boundary limits of $g$ as
		\begin{align}
			\lim_{s\to 1^-}g(s) = -r < 0, \qquad \lim_{s\to 0^+}g(s) = \lim_{s\to 0^+} -\phi'(s) - r > 0.
		\end{align}
		Since $g$ is continuous and strictly decreasing, there exists a unique $s_{r}
		\in(0,1)$ such that $g(s_{r}) = 0$, or equivalently,
		\begin{align}
			\label{eq: first order sr}r = (s_{r}-1) \phi'(s_{r}) - \phi(s_{r}).
		\end{align}
		Moreover, $g(s) > 0$ for $s < s_{r}$ and $g(s) < 0$ for $s > s_{r}$, so
		$f'(s) > 0$ for $s < s_{r}$ and $f'(s) < 0$ for $s > s_{r}$. Hence $f$ attains
		its unique maximum at $s_{r}$.

		An identical argument applied to $\phi_{n}$ (which is also convex and
		$C^{1}$) shows that $f_{n}(s) := \frac{-\phi_{n}(s)-sr}{1-s}$ has a unique maximizer
		$s_{n,r}\in (0,1)$ satisfying
		\begin{align}
			r = (s_{n,r}-1) \phi_{n}'(s_{n,r}) - \phi_{n}(s_{n,r}).
		\end{align}
		Define $g_{n}(s) := (s-1)\phi_{n}'(s) - \phi_{n}(s) - r$, so that
		$g_{n}(s_{n,r}) = 0$. 
		
		We claim that $s_{r}$ is a continuous function of $r$.
		Recall that $s_{r}$ is defined implicitly by $g(s_{r}) = 0$, i.e.,
		$y(s_{r}) = r$, where $y(s) := (s-1)\phi'(s) - \phi(s)$. We have already
		shown that $g(s) = y(s) - r$ is strictly decreasing and continuous on
		$(0,1)$. Hence $y$ is itself strictly decreasing and continuous on $(0,1)$,
		and $s_{r} = y^{-1}(r)$. Since the inverse of a continuous strictly
		monotone function is continuous, $s_{r}$ is continuous in $r$. An identical argument applies to $s_{n,r}$.

		We now show $s_{n,r}\to s_{r}$. Since $\{s_{n,r}\}\subset (0,1)$ is bounded,
		it suffices to show that every convergent subsequence has limit $s_{r}$. Let
		$s_{r,n_k}\to s^{*}$ for some $s^{*}\in[0,1]$. We show $s^{*}\in(0,1)$. If $s
		^{*} = 0$: for any $\delta\in(0,1)$, eventually $s_{r,n_k}<\delta$. Since
		$g_{n_k}$ is strictly decreasing and $g_{n_k}(s_{r,n_k})=0$, we have $g_{n_k}
		(\delta)<0$ for large $k$. Taking $k\to\infty$: $g(\delta)\le 0$. But this
		holds for all $\delta\in(0,1)$, contradicting $\lim_{s\to 0^+}g(s) > 0$. If
		$s^{*} = 1$: for any $\delta\in(0,1)$, eventually $s_{r,n_k}>\delta$. Since $g
		_{n_k}$ is strictly decreasing and $g_{n_k}(s_{r,n_k})=0$, we have
		$g_{n_k}(\delta)>0$ for large $k$. Taking $k\to\infty$: $g(\delta)\ge 0$. But
		this holds for all $\delta\in(0,1)$, contradicting $g(1^{-}) = -r < 0$. Hence
		$s^{*}\in(0,1)$. 
		
		Choose a compact set $K\subset(0,1)$ containing $s^{*}$ in
		its interior. For $k$ large enough, $s_{r,n_k}\in K$. Then we have
		\begin{align}
			|g(s^{*}) - g_{n_k}(s_{r,n_k})| & \le |g(s^{*}) - g(s_{r,n_k})| + |g(s_{r,n_k}) - g_{n_k}(s_{r,n_k})|  \\ & \le |g(s^{*}) - g(s_{r,n_k})| + \sup_{s\in K}|g(s) - g_{n}(s)|.
		\end{align}
		The first term vanishes as $k\to\infty$ by the continuity of $g$, and the
		second by the uniform convergence $g_{n}\to g$ on $K$. This follows from the fact that each $\phi_{n}$ is convex and
		$\phi_{n} \to \phi$ pointwise on $(0,1)$, we have $\phi_{n}\to\phi$
		uniformly on every compact subset of $(0,1)$ by~\cite[Theorem~10.8]{rockafellar}.
		Moreover, since $\phi$ is differentiable on $(0,1)$, $\phi_{n}'\to\phi'$ uniformly
		on every compact subset of $(0,1)$ by~\cite[Theorem~25.7]{rockafellar}. It
		follows that $g_{n} \to g$ uniformly on compact subsets of $(0,1)$.

		Since $g_{n_k}(s_{r,n_k}
		) = 0$, we obtain $g(s^{*}) = 0$. Since $g$ is strictly decreasing with a unique
		zero at $s_{r}$, we conclude $s^{*} = s_{r}$. Hence $s_{n,r}\to s_{r}$. This concludes the proof.
	\end{proof}

	We now combine the large-deviation estimate with the continuity control on the optimizing parameter to upper bound the error exponent of sequences of states. This is the core statement before we lift the argument to composite sets.

	\begin{boxproposition}
		\label{lem: Hoeffding bound sets} Let $\cH$ be a finite-dimensional Hilbert space
		and let $\rho^{(n)}, \sigma^{(n)}\in \density(\cH^{\ox n})$ be sequences of quantum
		states. Define
		\begin{align}
			\phi(s):= \lim_{n\to \infty}\frac{1}{n}\log \tr (\rho^{(n)})^{1-s}(\sigma^{(n)})^{s}
		\end{align}
		and assume that $\phi(s)$ exists and is finite for each $s\in(0,1)$. Assume further
		that $\phi(s)$ is strictly convex and $C^{1}$-continuous on $(0,1)$. Then
		for any
		$0< r < \sup_{\alpha\in(0,1)}D_{\Petz,\alpha}^{\reg}(\{\rho^{(n)}\}_{n}\|\{\sigma
		^{(n)}\}_{n})$,
		\begin{align}
			\limsup_{n\to \infty}-\frac{1}{n}\log \alpha_{n, r}(\rho^{(n)}\| \sigma^{(n)}) & \le \sup_{\alpha \in (0,1)}\frac{\alpha-1}{\alpha}\left(r - D_{\Petz,\alpha}^{\infty}(\{\rho^{(n)}\}_{n}\|\{\sigma^{(n)}\}_{n})\right). \label{NL3T}
		\end{align}
	\end{boxproposition}

	\begin{proof}
		For any test with $\tr \sigma^{(n)}T_{n} \le e^{-nr}$, we have
		\begin{align}
			\limsup_{n\to \infty}-\frac{1}{n}\log & \left[\tr e^{-nR}\rho^{(n)}(I-T_{n}) +\tr\sigma^{(n)}T_{n} \right]                         \\
			& \geq \limsup_{n\to \infty}-\frac{1}{n}\log \left[\tr e^{-nR}\rho^{(n)}(I-T_{n}) + e^{-nr}\right] \\
			& = \min \left\{r, R + \limsup_{n\to \infty}-\frac{1}{n}\log \tr \rho^{(n)}(I-T_{n}) \right\},
		\end{align}
		where the equality follows from~\cite[Lemma 45]{fang2025error}. Therefore, by
		Lemma~\ref{LH1}, we have
		\begin{align}
			\label{eq: limsup proof tmp1}\min \left\{r, R + \limsup_{n\to \infty}-\frac{1}{n}\log \tr \rho^{(n)}(I-T_{n}) \right\} \le\max_{s\in (0,1)}(1-s)R -\phi(s).
		\end{align}

		Since $\phi$ is convex with $\phi
		(0)=0$, for $0<s_{1}<s_{2}$ we have
		$\phi(s_{1}) = \phi\!\big(\frac{s_{1}}{s_{2}}s_{2} + (1-\frac{s_{1}}{s_{2}})\cdot
		0\big) \le \frac{s_{1}}{s_{2}}\phi(s_{2})$, so $\frac{\phi(s_{1})}{s_{1}}\le\frac{\phi(s_{2})}{s_{2}}$.
		Hence the slope $\frac{\phi(s)}{s}$ is non-decreasing in $s$, and $\lim_{s\to 0^+}\phi'(s) = \inf_{s\in(0,1)}\frac{\phi(s)}{s}$. Using
		$-\frac{\phi(s)}{s}= D_{\Petz,1-s}^{\infty}(\{\rho^{(n)}\}_{n}\|\{\sigma^{(n)}
		\}_{n})$
		and the substitution $\alpha = 1-s$, we obtain the relation
		\begin{align}
			\lim_{s\to 0^+}-\phi'(s) = \sup_{s\in(0,1)}\Big(-\frac{\phi(s)}{s}\Big) = \sup_{\alpha\in(0,1)}D_{\Petz,\alpha}^{\infty}(\{\rho^{(n)}\}_{n}\|\{\sigma^{(n)}\}_{n}).
		\end{align}
		Therefore, $r < \sup_{\alpha\in(0,1)}D_{\Petz,\alpha}^{\infty}(\{\rho^{(n)}\}_{n}\|\{\sigma^{(n)}\}_{n}) = \lim_{s\to 0^+} -\phi'(s)$.

		Let $r'<r$. From Lemma~\ref{lem: phi continuity}, we know that the
		optimization
		\begin{align}
			\max_{s\in (0,1)}\frac{-\phi(s)-sr'}{1-s}
		\end{align}
		has a unique maximizer $s_{r'}\in (0,1)$ such that $r' = (s_{r'}-1) \phi'(s_{r'}
		) - \phi(s_{r'})$. We set
		\begin{align}
			\label{eq: limsup proof tmp2}R_{r'}:=\frac{\phi(s_{r'})+r'}{1-s_{r'}}= -\phi'(s_{r'}).
		\end{align}
		It is clear that $R_{r'}\in (-\phi'(1), -\phi'(0))$. By the proof of Lemma~\ref{LH1},
		$\max_{s\in (0,1)}(1-s)R -\phi(s)$ is uniquely achieved at point $s$ such
		that $R = -\phi'(s)$. For $R = R_{r'}$, ~\eqref{eq: limsup proof tmp2}
		implies that the maximum is uniquely achieved at $s = s_{r'}$. That is,
		\begin{align}
			\max_{s\in (0,1)}(1-s)R_{r'}-\phi(s) = (1-s_{r'})R_{r'}-\phi(s_{r'}) = r'.
		\end{align}
		Then, by ~\eqref{eq: limsup proof tmp1} we have
		\begin{align}
			\min \left\{ r, R_{r'}+ \limsup_{n\to \infty}-\frac{1}{n}\log \tr \rho^{(n)}(I-T_{n}) \right\} \leq r'.
		\end{align}
		Then, we have
		\begin{align}
			R_{r'}+ \limsup_{n\to \infty}-\frac{1}{n}\log \tr \rho^{(n)}(I-T_{n}) \le r.
		\end{align}
		Otherwise, it will contradict to the assumption that $r' < r$. Thus, we have
		\begin{align}
			\limsup_{n\to \infty}-\frac{1}{n}\log \tr \rho^{(n)}(I-T_{n}) \le r - R_{r'}.
		\end{align}

		Note that by Lemma~\ref{lem: phi continuity}, $s_{r'}$ is continuous in $r'$. This implies the continuity of $R_{r'}$ in $r'$. Hence, sending $r'\to r^-$ gives
		\begin{align}
			\lim_{r'\to r^-}R_{r'}=\frac{\phi(s_{r})+r}{1-s_{r}}=: R_r.
		\end{align}
		This gives
		\begin{align}
			\limsup_{n\to \infty}- \frac{1}{n}\log \tr \rho^{(n)}(I-T_{n}) & \le r - R_{r}.
		\end{align}
		By direct calculation, the right-hand side gives
		\begin{align}
			r - R_{r}= \frac{-\phi(s_{r})-s_{r} r}{1-s_{r}}= \max_{s\in (0,1)}\frac{-\phi(s)-s r}{1-s},
		\end{align}
		where the second equality follows from the optimality of $s_{r}$. Then we
		have
		\begin{align}
			\limsup_{n\to \infty}- \frac{1}{n}\log \tr \rho^{(n)}(I-T_{n})\leq \max_{s\in (0,1)}\frac{-\phi(s)-s r}{1-s}.
		\end{align}
		This implies that
		\begin{align}
			\limsup_{n\to \infty}- \frac{1}{n}\log \alpha_{n, r}(\rho^{(n)}\| \sigma^{(n)})\leq \max_{s\in (0,1)}\frac{-\phi(s)-s r}{1-s}.
		\end{align}

		Finally, we show that
		\begin{align}
			\max_{s\in (0,1)}\frac{-\phi(s)-s r}{1-s} & = \sup_{\alpha \in (0,1)}\frac{\alpha-1}{\alpha}\left(r - D_{\Petz,\alpha}^{\infty}(\{\rho^{(n)}\}_{n}\|\{\sigma^{(n)}\}_{n})\right). \label{eq:phi D alpha limit}
		\end{align}
		By definition, $D_{\Petz,\alpha}
		^{\infty}(\{\rho^{(n)}\}_{n}\|\{\sigma^{(n)}\}_{n}) = \frac{1}{\alpha - 1}\phi(1-\alpha)$ for $\alpha\in(0,1)$.
		Writing $s = 1-\alpha$, we have $D_{\Petz,1-s}
		^{\infty}(\{\rho^{(n)}\}_{n}\|\{\sigma^{(n)}\}_{n}) = \frac{1}{-s}\phi(s) = -\frac{\phi(s)}{s}$. Thus, for each $s\in(0
		,1)$,
		\begin{align}
			\frac{-\phi(s)-s r}{1-s} & = \frac{s}{1-s}\!\left(-\frac{\phi(s)}{s}- r\right) = \frac{s}{1-s}\!\left(D_{\Petz,1-s}^{\infty}(\{\rho^{(n)}\}_{n}\|\{\sigma^{(n)}\}_{n}) - r\right). 
		\end{align}
		Under the substitution $\alpha = 1-s$, we have
		\begin{align}
			\frac{-\phi(s)-s r}{1-s}= \frac{\alpha-1}{\alpha}\!\left(r - D_{\Petz,\alpha}^{\infty}(\{\rho^{(n)}\}_{n}\|\{\sigma^{(n)}\}_{n})\right). \label{NL56}
		\end{align}
		Taking supremum over $\alpha \in (0,1)$ gives~\eqref{eq:phi D alpha limit}. This completes the proof.
	\end{proof}

	The next lemma addresses a fundamental obstacle in the composite setting: the minimizer of $D_{\Petz,\alpha}(\sA_n\|\sB_n)$ generally depends on $\alpha$, so no single pair of states simultaneously achieves the regularized divergence at every $\alpha$. The following construction overcomes this difficulty by averaging the $\alpha$-dependent minimizers over a grid that refines with $n$, yielding a \emph{universal} pair $(\rho^{(n)}, \sigma^{(n)})$ that is independent of any particular $\alpha$ yet asymptotically achieves $D_{\Petz,\alpha}^{\infty}(\sA\|\sB)$ for every $\alpha \in (0,1)$ at which the regularized divergence is continuous.
	
	\begin{boxlemma}
		[Universal tight state pair.]\label{lem: averaged state pair}
		Let $\cH$ be a finite-dimensional Hilbert space.
		Let $\sA:=\{\sA_{n}\}_{n\in \NN}$ and $\sB:=\{\sB_{n}\}_{n\in \NN}$ be sequences of convex
		sets $\sA_n, \sB_n \subseteq \density(\cH^{\ox n})$.
		Suppose that $\alpha \mapsto D_{\Petz,\alpha}^{\infty}(\sA\|\sB)$
		is continuous at $\alpha$ in $(0,1)$.
		For each $n$, define the averaged states
		\begin{align}
			\label{eq: choice of rhon and sigman}\rho^{(n)}:= \frac{1}{n+1}\sum_{j=0}^{n} \rho_{n}\!\left(\frac{j}{n}\right),\quad \sigma^{(n)}:= \frac{1}{n+1}\sum_{j=0}^{n} \sigma_{n}\!\left(\frac{j}{n}\right),
		\end{align}
		where, for each $\alpha \in (0,1)$, $(\rho_{n}(\alpha),\sigma_{n}(\alpha)) \in \sA_n \times \sB_n$ is an approximate minimizer satisfying
		\begin{align}
			{D}_{\Petz,\alpha}(\rho_{n}(\alpha)\|\sigma_{n}(\alpha)) \le D_{\Petz,\alpha}(\sA_{n}\|\sB_{n}) + 1.
		\end{align}
		Then $\rho^{(n)}\in \sA_{n}$,
		$\sigma^{(n)}\in \sB_{n}$, and for any $\alpha \in (0,1)$,
		\begin{align}
			{D}^\infty_{\Petz,\alpha}(\{\rho^{(n)}\}\| \{\sigma^{(n)}\}) ={D}_{\Petz,\alpha}^{\infty}(\sA\| \sB).\label{KSD}
		\end{align}
	\end{boxlemma}

	\begin{proof}
		For each $\alpha \in (0,1)$ and $n \in \NN$, the approximate minimizer
		$(\rho_{n}(\alpha),\sigma_{n}(\alpha))$ exists because $D_{\Petz,\alpha}(\sA_n\|\sB_n)$ is defined as an infimum, so a pair achieving the infimum to within additive error~$1$ can always be found.
		Since $\sA_{n}$ and $\sB_{n}$ are convex, we have that $\rho^{(n)}\in \sA_{n}$ and $\sigma^{(n)}\in \sB_{n}$. This implies that
		\begin{align}
			\frac{1}{n}D_{\Petz,\alpha}(\rho^{(n)}\|\sigma^{(n)})
			\ge \frac{1}{n}D_{\Petz,\alpha}(\sA_{n}\| \sB_{n}).
		\end{align}
		Taking $\liminf_{n\to\infty}$ yields
		\begin{align} 
            \liminf_{n\to \infty}\frac{1}{n}D_{\Petz,\alpha}(\rho^{(n)}\|
		\sigma^{(n)}) \ge D_{\Petz,\alpha}^{\infty}(\sA\| \sB). \label{eq: avg lower}
        \end{align}

		For the reverse direction, we use Lieb's concavity theorem, which states that the map
		$(A,B) \mapsto \tr A^{\alpha} B^{1-\alpha}$ is jointly concave for
		$\alpha \in (0,1)$. This gives
		\begin{align}
			\tr (\rho^{(n)})^{\alpha} (\sigma^{(n)})^{1-\alpha} & \ge \frac{1}{(n+1)^{2}}\sum_{j=0}^{n} \sum_{j'=0}^{n} \tr \rho_{n}\!\left(\tfrac{j}{n}\right)^{\alpha} \sigma_{n}\!\left(\tfrac{j'}{n}\right)^{1-\alpha}. \label{eq: Lieb double sum}
		\end{align}
		Every term in the double sum is nonnegative, so for any
		$j'' \in \{0,\ldots,n\}$ we can keep only the diagonal term
		$j = j' = j''$:
		\begin{align}
			\tr (\rho^{(n)})^{\alpha} (\sigma^{(n)})^{1-\alpha}
			& \ge \frac{1}{(n+1)^{2}}\tr \rho_{n}\!\left(\tfrac{j''}{n}\right)^{\alpha} \sigma_{n}\!\left(\tfrac{j''}{n}\right)^{1-\alpha}. \label{eq: keep one term}
		\end{align}
		This implies that
		\begin{align}
			\frac{1}{n}D_{\Petz,\alpha}(\rho^{(n)}\|\sigma^{(n)}) & \le \frac{1}{n}D_{\Petz,\alpha}\!\left(\rho_{n}\!\left(\tfrac{j''}{n}\right) \middle\| \sigma_{n}\!\left(\tfrac{j''}{n}\right)\right) + \frac{2\log(n+1)}{n(1-\alpha)}. \label{eq: after log}
		\end{align}
		Now choose $j''$ to be the smallest index in
		$\{0,\ldots,n\}$ such that $j''/n \geq \alpha$. Then
		$\alpha \le j''/n \le \alpha + 1/n$. Since $D_{\Petz,\alpha}(\rho\|\sigma)$ is
		non-decreasing in $\alpha$ for fixed states $\rho, \sigma$, we obtain
		\begin{align}
			D_{\Petz,\alpha}\!\left(\rho_{n}\!\left(\tfrac{j''}{n}\right) \middle\| \sigma_{n}\!\left(\tfrac{j''}{n}\right)\right)
			& \leq D_{\Petz,j''/n}\!\left(\rho_{n}\!\left(\tfrac{j''}{n}\right) \middle\| \sigma_{n}\!\left(\tfrac{j''}{n}\right)\right)\\
			& \leq D_{\Petz,j''/n}(\sA_{n}\|\sB_{n}) + 1\\
            & \leq  D_{\Petz,\alpha+1/n}(\sA_{n}\|\sB_{n}) + 1. \label{eq: mono and min}
		\end{align}
		where the first inequality uses the monotonicity $\alpha \le j''/n$, the second is the approximate minimizer condition, and the third uses $j''/n \leq \alpha + 1/n$. Combining~\eqref{eq: after log}
		and~\eqref{eq: mono and min}, we get
		\begin{align}
			\frac{1}{n}D_{\Petz,\alpha}(\rho^{(n)}\|\sigma^{(n)}) & \le \frac{1}{n}D_{\Petz,\alpha + 1/n}(\sA_{n}\|\sB_{n}) + \frac{2\log(n+1)}{n(1-\alpha)} + \frac{1}{n}. \label{eq: upper before limsup}
		\end{align}
        Fix any $\beta \in (\alpha,1)$. For all $n > 1/(\beta - \alpha)$, we have $\alpha + 1/n < \beta$. Since $D_{\Petz,\alpha}$ is non-decreasing in $\alpha$ on $(0,1)$, it follows that
		\begin{align}
			\frac{1}{n}D_{\Petz,\alpha+1/n}(\sA_n\|\sB_n) \leq \frac{1}{n}D_{\Petz,\beta}(\sA_n\|\sB_n).
		\end{align}
		Substituting into~\eqref{eq: upper before limsup} and taking $\limsup_{n\to\infty}$:
		\begin{align}
			\limsup_{n\to \infty} \frac{1}{n}D_{\Petz,\alpha}(\rho^{(n)}\|\sigma^{(n)}) \leq D_{\Petz,\beta}^{\infty}(\sA\|\sB).
		\end{align}
		Since this holds for every $\beta > \alpha$, taking $\beta \to \alpha^+$ and using the assumed continuity of $\alpha \mapsto D_{\Petz,\alpha}^{\infty}(\sA\|\sB)$ at $\alpha$ yields
		\begin{align}
			\limsup_{n\to \infty} \frac{1}{n}D_{\Petz,\alpha}(\rho^{(n)}\|\sigma^{(n)}) \leq D_{\Petz,\alpha}^{\infty}(\sA\|\sB).
		\end{align}
		Together with~\eqref{eq: avg lower}, this establishes~\eqref{KSD}.
	\end{proof}

	Combining Proposition~\ref{lem: Hoeffding bound sets} and Lemma~\ref{lem: averaged state pair}, we obtain the general upper bound for convex composite hypotheses. Conceptually, the theorem shows that a single asymptotically tight pair of states is enough to control the exponent for the entire family.

	\begin{boxtheorem}[Quantum Hoeffding upper bound via $D_{\Petz,\alpha}$.]
		\label{thm: quantum Hoeffding upper bound}
		Let $\cH$ be a finite-dimensional Hilbert space.
		Let $\sA = \{\sA_{n}\}_{n\in\NN}$ and $\sB = \{\sB_{n}\}_{n\in\NN}$ be
		sequences of sets of quantum states such that $\sA_n, \sB_n \subseteq \density(\cH^{\ox n})$ for every $n\in \NN$.
		Suppose that 
		\begin{itemize}
			\item Both $\sA$ and $\sB$ satisfy the assumption (C1) convexity;
			\item $(\sA, \sB)$ satisfy the assumptions (C6)  strict-concavity and (C7) continuity.
		\end{itemize}
		Then for any $0 < r < \sup_{\alpha \in (0,1)} D_{\Petz,\alpha}^{\infty}(\sA\|\sB)$,
		\begin{align}
			\limsup_{n\to \infty}-\frac{1}{n}\log \alpha_{n,r}(\sA_n\|\sB_n) & \le \sup_{\alpha \in (0,1)}\frac{\alpha-1}{\alpha}\bigg(r -{D}_{\Petz,\alpha}^{\infty}(\sA\| \sB) \bigg). \label{eq: general Hoeffding bound}
		\end{align}
	\end{boxtheorem}

	\begin{proof}
		Given the required assumptions, we can apply Lemma~\ref{lem: averaged state pair} to obtain a sequence of states $\rho^{(n)} \in \sA_n$ and $\sigma^{(n)} \in \sB_n$ such that $D_{\Petz,\alpha}^{\infty}(\{\rho^{(n)}\}\|\{\sigma^{(n)}\}) = D_{\Petz,\alpha}^{\infty}(\sA\|\sB)$ for every $\alpha \in (0,1)$. Applying Proposition~\ref{lem: Hoeffding bound sets} to the pair of sequences $\{\rho^{(n)}\}$ and $\{\sigma^{(n)}\}$, we get 
		\begin{align}
			\limsup_{n\to \infty}-\frac{1}{n}\log \alpha_{n,r}(\rho^{(n)}\|\sigma^{(n)}) & \le \sup_{\alpha \in (0,1)}\frac{\alpha-1}{\alpha}\bigg(r -{D}_{\Petz,\alpha}^{\infty}(\sA\| \sB) \bigg).
		\end{align}
		As $\rho^{(n)} \in \sA_n$ and $\sigma^{(n)} \in \sB_n$, we have $\alpha_{n,r}(\sA_n\|\sB_n) \geq \alpha_{n,r}(\rho^{(n)}\|\sigma^{(n)})$ by definition. Taking this into the above inequality gives the asserted result.
	\end{proof}

	\section{Matching bounds and operational interpretations}
	\label{sec: matching bounds}

	The lower and upper bounds derived in the previous section are governed by distinct R\'enyi divergences and need not coincide in general. In this section, we identify conditions under which the two bounds match. We first establish that pinching with respect to a family of orthogonal projections is equivalent to twirling over the associated diagonal unitary group. Building on this equivalence, we then show that, under suitable group compatibility assumptions, the composite Hoeffding exponent is exactly characterized by the regularized reverse sandwiched \Renyi divergence. Specializing further to a particular composite hypothesis testing setting, the regularization can be removed and the exponent reduces to a single-letter expression in the reverse sandwiched \Renyi divergence, thereby providing an operational interpretation of this divergence. This further leads to an operational interpretation of the reverse quantum relative entropy in the Stein regime.

	\subsection{Equivalence of twirling and pinching}

	Fix a family of orthogonal projections $\{P_i\}_{i=1}^{m}$ on $\cH$ satisfying $\sum_{i=1}^{m} P_i = I$. For each multi-index $\vec{j} = (j_1,\dots,j_n) \in \{1,\dots,m\}^{n}$, let $P_{\vec{j}} := P_{j_1}\otimes\cdots\otimes P_{j_n}$ denote the associated product projection. For each $k \in \{1,\dots,m\}$, let $N_k(\vec{j}) := \bigl|\{\ell : j_\ell = k\}\bigr|$ denote the number of coordinates of $\vec{j}$ equal to~$k$, and let $N(\vec{j}) := (N_1(\vec{j}),\dots,N_m(\vec{j}))$ be its type. Let 
    \begin{align}\label{eq: type projector}
        \Pi_{t} := \sum_{\vec{j} \,:\, N(\vec{j}) = t} P_{\vec{j}},
    \end{align}
    be the projector onto the type class~$t$. We denote the set of types of length~$n$ by $\mathcal{T}_{n,m}$, so that
    \begin{align} 
	|\mathcal{T}_{n,m}| = \binom{n+m-1}{m-1} \quad \text{and} \quad \sum_{t \in \mathcal{T}_{n,m}} \Pi_t = I_{\cH^{\otimes n}}.
    \end{align}
	We define the pinching map with respect to the projections $\{P_i\}_{i=1}^{m}$ as
    \begin{align}\label{eq: pinching def}
        \cP^n(X) := \sum_{t \in \mathcal{T}_{n,m}} \Pi_t X \Pi_t,
    \end{align}

    For the projections $\{P_i\}_{i=1}^{m}$, we consider the corresponding compact group of unitaries
    \begin{align}\label{eq: group def}
        \group:= \left\{g = \sum_{j=1}^{m} e^{i\theta_j} P_j : \theta_i \in [0,2\pi) \right\},
    \end{align}
    equipped with the normalized Haar measure $\mathrm{d}g = \prod_{j=1}^{m}\frac{\mathrm{d}\theta_j}{2\pi}$. For each $n \in \NN$, we let $\group$ act on $\cH^{\otimes n}$ through the representation $g \mapsto g^{\otimes n}$, and write the associated twirling:
    \begin{align}\label{eq: twirling def}
        \cT_{\group}^{n}(X) := \int_{\group} g^{\otimes n}\, X\, (g^{\dagger})^{\otimes n}\, \mathrm{d}g.
    \end{align}

	The following lemma shows that these two constructions are in fact identical: pinching with respect to the projections $\{P_i\}$ is exactly the same as twirling over the diagonal unitary group $\group$. This equivalence gives the otherwise algebraic pinching map a simple ``operational interpretation''. In essence, pinching is just the averaging effect of uncertainty in the unitary phases.

    \begin{boxlemma}\label{lem: twirling equals pinching}
        Let $\{P_i\}_{i=1}^{m}$ be a family of orthogonal projections on $\cH$ with $\sum_{i=1}^{m} P_i = I$. Let  $\cP^n$ and $\cT_{\group}^{n}$ be defined as in~\eqref{eq: pinching def} and~\eqref{eq: twirling def}, respectively. Then, it holds that
        $\cP^n = \cT_{\group}^{n}$.
    \end{boxlemma}

    \begin{proof}
        For $g = g_\theta \in \group$ we have $g_\theta^{\otimes n} = \sum_{\vec{j} \in \{1,\dots,m\}^{n}} e^{i\sum_{\ell=1}^{n}\theta_{j_\ell}}\, P_{\vec{j}}$, and therefore
        \begin{align}
            g_\theta^{\otimes n}\, X\, (g_\theta^{\dagger})^{\otimes n}
            = \sum_{\vec{j},\vec{k} \in \{1,\dots,m\}^{n}} e^{i\sum_{\ell=1}^{n}(\theta_{j_\ell} - \theta_{k_\ell})}\, P_{\vec{j}}\, X\, P_{\vec{k}}.
        \end{align}
        Regrouping the sum $\sum_{\ell=1}^{n}\theta_{j_\ell}$ according to the value of each coordinate gives $\sum_{\ell=1}^{n}\theta_{j_\ell} = \sum_{r=1}^{m}\theta_r\, N_r(\vec{j})$, and hence
        \begin{align}\label{eq: exponent rewritten}
            \sum_{\ell=1}^{n}\bigl(\theta_{j_\ell} - \theta_{k_\ell}\bigr)
            = \sum_{r=1}^{m}\theta_{r}\bigl(N_{r}(\vec{j}) - N_{r}(\vec{k})\bigr).
        \end{align}
        Substituting~\eqref{eq: exponent rewritten} into the phase, the integrand factorizes as $\prod_{r=1}^{m} e^{i\theta_r(N_r(\vec{j}) - N_r(\vec{k}))}$. Applying Fubini's theorem together with the one-dimensional orthogonality relation $\frac{1}{2\pi}\int_{0}^{2\pi} e^{i\theta_r a_r}\,\mathrm{d}\theta_r = \delta_{a_r,0}$ for $a_r \in \ZZ$, we obtain
        \begin{align}
            \int_{[0,2\pi)^{m}} e^{i\sum_{r=1}^{m}\theta_{r}(N_{r}(\vec{j}) - N_{r}(\vec{k}))}\, \prod_{r=1}^{m}\frac{\mathrm{d}\theta_r}{2\pi}
            = \prod_{r=1}^{m} \delta_{N_{r}(\vec{j}),\, N_{r}(\vec{k})}
            = \delta_{N(\vec{j}),\, N(\vec{k})},
        \end{align}
        which vanishes unless $\vec{j}$ and $\vec{k}$ have the same type. Grouping the surviving pairs by their common type $t$ then yields
        \begin{align}
            \cT_{\group}^{n}(X)
            = \sum_{t \in \mathcal{T}_{n,m}} \ \ \sum_{\vec{j},\vec{k}\,:\, N(\vec{j}) = N(\vec{k}) = t}\!\! P_{\vec{j}}\, X\, P_{\vec{k}}
            = \sum_{t \in \mathcal{T}_{n,m}} \Pi_{t}\, X\, \Pi_{t},
        \end{align}
        where the last step uses the definition~\eqref{eq: type projector} of~$\Pi_{t}$. This completes the proof.
    \end{proof}

	The next lemma explains why this identification is useful for the Hoeffding problem: when the null hypothesis is invariant under the group action, twirling the alternative converts the Petz \Renyi divergence into the reverse sandwiched \Renyi divergence. This conversion is precisely the mechanism that will align the lower and upper bounds in the application below.

	\begin{boxlemma}
		\label{lem: Petz RSand twirling}
		Let $\cH$ be a finite-dimensional Hilbert space, let $\{P_i\}_{i=1}^{m}$
		be orthogonal projections on $\cH$ with $\sum_{i} P_i = I$. Let
		$\cP^n$, $\group$ and $\cT_{\group}^{n}$ defined as in~\eqref{eq: pinching def},~\eqref{eq: group def} and~\eqref{eq: twirling def}, respectively.
		Let $\sA = \{\sA_{n}\}_{n\in\NN}$ and $\sB = \{\sB_{n}\}_{n\in\NN}$ be
		sequences with each $\sA_n, \sB_n \subseteq \density(\cH^{\ox n})$. Denote the sequence of twirled sets by $\cT_{\group}(\sB):=\{\cT_{\group}^{n}(\sB_n)\}_{n\in\NN}$. Suppose that every $\rho_n \in \sA_n$ has the form $\rho_n=\sum_{t \in \cT_{n,m}} s_t \Pi_t$, where $s_t$ may depend on $\rho_n$.
		Then for any $\alpha \in (0,1) \cup (1,+\infty)$, 
		\begin{align}\label{eq: Petz RSand infty}
			{D}_{\Petz,\alpha}^{\infty}\!\left(\sA\,\middle\|\,\cT_{\group}(\sB)\right)
			 = D_{\RSand,\alpha}^{\infty}\!\left(\sA\,\middle\|\,\sB\right).
		\end{align}
	\end{boxlemma}
	\begin{proof}
		By the definition of the reverse sandwiched R\'enyi divergence and Petz \Renyi divergence, the asserted result is equivalent to 
		\begin{align}
			{D}_{\Petz,\alpha}^{\infty}\!\left(\cT_{\group}(\sB)\middle\|\,\sA\ \right)
			 = D_{\Sand,\alpha}^{\infty}\!\left(\sB\,\middle\|\,\sA\right).
		\end{align}
		
		For any $\rho_n \in \sA_n$ and $\sigma_n \in \sB_n$, we have 
		\begin{align}
			D_{\Petz,\alpha}(\cT_{\group}^{n}(\sigma_n)\|\rho_n) & = D_{\Sand,\alpha}(\cT_{\group}^{n}(\sigma_n)\|\rho_n) \leq D_{\Sand,\alpha}(\sigma_n\|\rho_n),
		\end{align}
		where the equality follows because $\cT_{\group}^{n}(\sigma_n)$ commutes with $\rho_n =\sum_{t\in \cT_{n,m}} s_{t} \Pi_{t}$, and the inequality is the data-processing inequality under the twirling map~\cite[Proposition~14]{muller2013quantum} (note that $\rho_n = \cT_{\group}^n(\rho_n)$). Optimizing over $\rho_n \in \sA_n$, $\sigma_n \in \sB_n$ and taking the regularized limit yields
		\begin{align}
			{D}_{\Petz,\alpha}^{\infty}\!\left(\cT_{\group}(\sB)\middle\|\,\sA\ \right)
			 \leq D_{\Sand,\alpha}^{\infty}\!\left(\sB\,\middle\|\,\sA\right).
		\end{align}

		For the reverse direction, fix any $\rho_n \in \sA_n$ and $\sigma_n \in \sB_n$. We have
		\begin{align}
			D_{\Petz,\alpha}(\cT_{\group}^{n}(\sigma_n)\|\rho_n) & \geq D_{\Sand,\alpha}(\cT_{\group}^{n}(\sigma_n)\|\rho_n)\\
			& \geq D_{\Sand,\alpha}(\sigma_n\|\rho_n) -2\log |\cT_{n,m}|\\
			& \geq D_{\Sand,\alpha}\!\left(\sB_{n}\middle\|\,\sA_n\ \right) - 2\log |\cT_{n,m}|,
		\end{align}
		where the first inequality is $D_{\Petz,\alpha} \geq D_{\Sand,\alpha}$, the second follows from~\cite[Lemma~3]{hayashi2016correlation}, and the third uses $\sigma_n \in \sB_n$ and $\rho_n \in \sA_n$.
		Optimizing over $\rho_n \in \sA_n$ and $\sigma_n \in \sB_n$ and taking the regularized limit yields
		\begin{align}
			{D}_{\Petz,\alpha}^{\infty}\!\left(\cT_{\group}(\sB)\middle\|\,\sA\ \right) \geq D_{\Sand,\alpha}^{\infty}\!\left(\sB\,\middle\|\,\sA\right).
		\end{align}
		This completes the proof.
	\end{proof}

	\subsection{Matching bounds}

	Combining Lemma~\ref{lem: Petz RSand twirling} with the bounds in the previous section, we identify a set of conditions under which the reverse sandwiched \Renyi divergence governs the Hoeffding exponent.

	\begin{boxassumption}
		Let $\group$ be the compact group defined in~\eqref{eq: group def}. For a sequence of sets of quantum states $\sC = \{\sC_{n}\}_{n\in \NN}$ with each $\sC_{n} \subseteq \density(\cH^{\ox n})$, we introduce the following assumptions.
		\begin{center}
			\renewcommand{\arraystretch}{1.5}
			\begin{tabular}{@{} c l p{10.5cm} @{}}
				\toprule \textbf{Label} & \textbf{Name}                   & \textbf{Description}                                                                                                               \\
				\midrule 
				(C1$''$)                    & ${\group}$-convexity        & For any $n \in \NN$, the twirled set $\cT_{\group}^{n}(\sC_{n})$ is convex.                                                         \\
				(C8)                   & ${\group}$-closedness       & For any $n \in \NN$, $g^{\ox n}(\rho_{n})(g^\dagger)^{\ox n} \in \sC_n$, $\forall \rho_{n} \in \sC_{n}, \forall g \in \group$.                           \\
				(C8$'$)           & ${\group}$-block-constantness       
				& For any $n \in \NN$, 
$\rho_{n}$				has the form $\sum_{t \in \cT_{n,m}} s_t \Pi_t$, $\forall \rho_{n} \in \sC_{n}$.                        \\                	
				\bottomrule
			\end{tabular}
		\end{center}
	\end{boxassumption}

	Note that (C1$''$) $\group$-convexity is strictly weaker than the standard convexity (C1), since it only requires convexity after twirling; by the linearity of twirling, (C1) automatically implies (C1$''$).
	
	With these assumptions, the following theorem shows that the composite Hoeffding exponent is governed by the \emph{reverse sandwiched} \Renyi divergence rather than the Petz \Renyi divergence familiar from the simple i.i.d.\ setting. This distinction is significant both conceptually and quantitatively: conceptually, it endows the reverse sandwiched divergence with an exact operational role in hypothesis testing; quantitatively, the two divergences can differ strictly for non-commuting states (see Example~\ref{ex: divergence comparison}), so the resulting exponent can depart substantially from its i.i.d.\ counterpart.

	\begin{boxtheorem}[Quantum Hoeffding bound via $D_{\RSand,\alpha}$.]
		\label{thm: matching bounds}
		Let $\cH$ be a finite-dimensional Hilbert space and $\group$ be a compact group of unitaries as defined in~\eqref{eq: group def}.
		Let $\sA = \{\sA_{n}\}_{n\in \NN}$ and $\sB = \{\sB_{n}\}_{n\in \NN}$ be two sequences of sets of quantum states with each $\sA_{n}, \sB_{n} \subseteq \density(\cH^{\ox n})$. Denote $\cT_{\group}(\sB):=\{\cT_{\group}^{n}(\sB_n)\}_{n\in\NN}$. 
		Suppose that
		\begin{itemize}
			\item $\sA$ satisfies assumptions (C1) convexity, (C5) $\perm$-closedness, and (C8$'$) $\group$-block-constantness;

			\item $\sB$ satisfies assumptions (C1$''$) $\group$-convexity, (C5$'$) $\perm$-invariance, and (C8) $\group$-closedness;

			\item $(\sA, \cT_\group(\sB))$ satisfy assumptions (C6) strict-concavity and (C7) continuity.
		\end{itemize}
		Then for any $0 < r < \sup_{\alpha \in (0,1)}D_{\RSand,\alpha}^{\reg}(\sA\|\sB)$,
		\begin{align}\label{eq: exact Hoeffding RSand}
			\lim_{n\to \infty}-\frac{1}{n}\log \alpha_{n, r}(\sA_{n}\| \sB_{n})
			= \sup_{\alpha \in (0,1)}\frac{\alpha-1}{\alpha}\bigg(r - D_{\RSand,\alpha}^{\infty}(\sA\|\sB)\bigg).
		\end{align}
	\end{boxtheorem}
	\begin{proof}
		The lower bound follows from Theorem~\ref{thm: quantum Hoeffding lower bound}.
		We now turn to the upper bound. Applying Theorem~\ref{thm: quantum Hoeffding upper bound} to the sequences $\sA$ and $\cT_{\group}(\sB)$, we have, for any $0 < r < \sup_{\alpha \in (0,1)} D_{\Petz,\alpha}^{\infty}(\sA\|\cT_{\group}(\sB))$,
		\begin{align}
			\limsup_{n\to \infty}-\frac{1}{n}\log \alpha_{n, r}(\sA_n\|\cT_{\group}^{n}(\sB_{n})) & \le \sup_{\alpha \in (0,1)}\frac{\alpha-1}{\alpha}\bigg(r -{D}_{\Petz,\alpha}^{\infty}(\sA\|\cT_{\group}(\sB)) \bigg). \label{eq: upper bound coro}
		\end{align}
		By Lemma~\ref{lem: Petz RSand twirling}, ${D}_{\Petz,\alpha}^{\infty}(\sA\|\cT_{\group}(\sB)) = D_{\RSand,\alpha}^{\infty}(\sA\|\sB)$. Moreover, since $\sB_n$ is closed under the group action, $\cT_{\group}^{n}(\sB_{n}) \subseteq \conv(\sB_n)$, and hence
		\begin{align}
			\alpha_{n, r}(\sA_{n}\|\sB_{n}) = \alpha_{n, r}(\sA_{n}\|\conv(\sB_{n})) \geq \alpha_{n, r}(\sA_{n}\|\cT_{\group}^{n}(\sB_{n})).
		\end{align}
		Substituting these into~\eqref{eq: upper bound coro} establishes the asserted upper bound and concludes the proof.
	\end{proof}

	\subsection{Operational interpretations}

	The assumptions in Theorem~\ref{thm: matching bounds} may appear rather strong at first sight, and the regularized expression on the right-hand side might suggest that the resulting exponent is difficult to evaluate. We now demonstrate that the framework is indeed useful by exhibiting a physically natural composite hypothesis-testing problem in which all of these assumptions are satisfied and the regularization can be removed. Consequently, the error exponent admits a single-letter expression in terms of the reverse sandwiched R\'enyi divergence, providing an operational interpretation of this divergence between a pair of quantum states. This then leads to operational interpretation of the reverse quantum relative entropy as well. Below, we discuss two physical models that fit into this framework: independent phase noise in the energy eigenbasis, and free time evolution under a Hamiltonian.
	
	\subsubsection*{Independent phase noise in energy eigenbasis}

	Following the setting outlined in the introduction, consider a quantum system with Hamiltonian $H = \sum_j E_j \ket{E_j}\bra{E_j}$, where $\{\ket{E_j}\}$ denotes the energy eigenbasis with eigenvalues $\{E_j\}$. The null hypothesis is the thermal equilibrium (Gibbs) state at inverse temperature $\beta$,
	\begin{align} 
		\rho = e^{-\beta H}/\tr[e^{-\beta H}],
	\end{align}
	which is by construction diagonal in the energy eigenbasis. The alternative hypothesis is a non-equilibrium probe state $\sigma$, which may carry off-diagonal coherences between energy levels. Before reaching the tester, however, each energy level acquires an unknown phase noise. This models dephasing in the energy eigenbasis, arising for instance from imprecise waiting times, clock misalignment, or fluctuating environmental couplings to individual energy levels.
	The state delivered to the tester therefore takes the form $g\, \sigma\, g^\dagger$ for some unknown $g \in \group$, with
	\begin{align}\label{eq: group def corollary}
        \group:= \left\{g = \sum_{j=1}^{m} e^{i\theta_j} \ket{E_j}\bra{E_j} : \theta_j \in [0,2\pi) \right\},
	\end{align}
	and the tester's task is to distinguish $\rho^{\otimes n}$ from such phase-corrupted copies of $\sigma^{\otimes n}$. We accordingly define the sequences of sets
	\begin{align}
		\sA_n &:= \{\rho^{\otimes n}\}, \quad \text{and}\quad
		\sB_n := \left\{g^{\otimes n}\sigma^{\otimes n}(g^{\otimes n})^\dagger : g \in \group\right\}, \label{eq: orbit alt set}
	\end{align}
	where we note that $\rho$ is invariant under the action of $\group$ by definition.

	Free time evolution under $H$ is a special case of this noise model: the unitary $e^{-iHt} = \sum_j e^{-iE_j t}\ket{E_j}\bra{E_j}$ applies correlated phases $\theta_j = -E_j t$ governed by the single parameter $t$, a setting we will discuss in detail later. The group $\group$ generalizes this by allowing each energy level to accrue an independent phase, thereby capturing the most general dephasing in the energy eigenbasis.
	The following theorem delivers the main operational payoff of this paper: the optimal Hoeffding exponent for discriminating the thermal state from its phase-corrupted alternative is given exactly by the single-letter reverse sandwiched \Renyi divergence. We emphasize that the alternative hypothesis is composite i.i.d.\ and therefore not tensor-stable, so the framework of~\cite{fang2025error} does not apply.

	\begin{boxtheorem}[Operational meaning of $D_{\RSand,\alpha}$.]
		\label{thm: RSand operational meaning}
		Let $\rho, \sigma \in \density(\cH)$ with $\rho > 0$ and $\sigma > 0$, and let $\sA_n$, $\sB_n$ be as defined in~\eqref{eq: orbit alt set}.
		Then for any $0 < r < D_{\Rev}(\rho\|\sigma)$,
		\begin{align}\label{eq: RSand operational meaning}
			\lim_{n\to\infty}-\frac{1}{n}\log \alpha_{n,r}(\sA_n\|\sB_n)
			= \sup_{\alpha \in (0,1)} \frac{\alpha-1}{\alpha}\big(r - D_{\RSand,\alpha}(\rho\|\sigma)\big).
		\end{align}
	\end{boxtheorem}

	\begin{proof}
		We verify the assumptions of Theorem~\ref{thm: matching bounds} for $\sA_n$ and $\sB_n$.

		\medskip
		\textit{1) Verification for $\sA$.} The singleton $\sA_n = \{\rho^{\otimes n}\}$ is convex, permutation-invariant, and hence permutation-closed. Since $\rho$ is diagonal in the basis $\{|E_j\rangle\}$, $\rho^{\ox n}$ is $\group$-block-constant.

		\medskip
		\textit{2) Verification for $\sB$.} For (C5$'$), every state in $\sB_n$ has the product form $g^{\otimes n}\sigma^{\otimes n}(g^{\otimes n})^\dagger$ and is therefore permutation-invariant. For (C1$''$), we have $\cT_\group^n(g^{\otimes n}\sigma^{\otimes n}(g^{\otimes n})^\dagger) = \cT_\group^n(\sigma^{\otimes n})$ for all $g \in \group$, so $\cT_\group^n(\sB_n) = \{\cT_\group^n(\sigma^{\otimes n})\}$ is a singleton and hence convex. For $\group$-closedness~(C8), note that $\sB_n$ is the orbit of $\sigma^{\otimes n}$ under the group action, so it is closed under that action by construction.
		
		\medskip
		\textit{3) Verification for $(\sA, \cT_\group(\sB))$.} We next compute the regularized divergence. Since $D_{\RSand,\alpha}(\rho\|\sigma)$ is unitarily invariant and $\rho$ commutes with every $g \in \group$, 
		\begin{align}
			D_{\RSand,\alpha}(\sA_n\|\sB_n) & = \inf_{g \in \group} D_{\RSand,\alpha}\big(\rho^{\otimes n}\|g^{\ox n} \sigma^{\ox n} (g^\dagger)^{\ox n}\big)
			 = n\, D_{\RSand,\alpha}(\rho\|\sigma).
		\end{align}
		This gives
		\begin{align}\label{eq: operational meaning proof tmp1}
			D_{\RSand,\alpha}^\infty(\sA\|\sB) = D_{\RSand,\alpha}(\rho\|\sigma) = -\log Q_{\Sand,1-\alpha}(\sigma\|\rho).
		\end{align}
		By Lemma~\ref{lem: Petz RSand twirling}, 
		\begin{align}\label{eq: operational meaning proof tmp2}
			(1-\alpha){D}_{\Petz,\alpha}^{\infty}\!\left(\sA\middle\|\,\cT_{\group}(\sB)\right) = (1-\alpha) D_{\RSand,\alpha}^{\infty}\!\left(\sA\middle\|\,\sB\right) = -\log Q_{\Sand,1-\alpha}(\sigma\|\rho).
		\end{align}
		Thus verifying assumptions (C6) and (C7) for the pair $(\sA, \cT_\group(\sB))$ reduces to showing that the single-letter function
		$\alpha \mapsto -\log Q_{\Sand,1-\alpha}(\sigma\|\rho)$
		is strictly concave and $C^1$-continuous on $(0,1)$.
		Setting $\beta := 1-\alpha \in (0,1)$, it is equivalent to establish that 
		\begin{align} 
			\beta \mapsto g(\beta) := \log Q_{\Sand,\beta}(\sigma\|\rho)
		\end{align}
		is strictly convex and $C^1$-continuous on $(0,1)$. We verify these two properties in turn.

		\bigskip 
		\textit{3.1) Verification for $C^1$-continuity.} We first show that $g(\beta)$ is real-analytic. Recall that a scalar or matrix-valued function on an open interval is called real-analytic if around every point it admits a convergent power-series expansion. We will use the standard facts that sums, products, and compositions of real-analytic maps are again real-analytic, that the scalar logarithm is real-analytic on $(0,\infty)$, and that the principal matrix logarithm and the matrix exponential are real-analytic on the cone of positive definite matrices.

		Since both $\rho$ and $\sigma$ are full rank, for every $\beta \in (0,1)$ the operator
		\begin{align}
			X(\beta) := \sigma^{1/2}\rho^{(1-\beta)/\beta}\sigma^{1/2}
		\end{align}
		is positive definite on $\cH$. Moreover,
		\begin{align}
			Q_{\Sand,\beta}(\sigma\|\rho) = \tr\big[X(\beta)^{\beta}\big],
		\end{align}
		because $\rho^{\frac{1-\beta}{2\beta}}\sigma\rho^{\frac{1-\beta}{2\beta}}$ and $\sigma^{1/2}\rho^{\frac{1-\beta}{\beta}}\sigma^{1/2}$ have the same non-zero eigenvalues.

		Now the scalar map $\beta \mapsto \frac{1-\beta}{\beta}$ is real-analytic on $(0,1)$, and since $\rho>0$ we may write
		\begin{align}
			\rho^{(1-\beta)/\beta} = \exp\!\left(\frac{1-\beta}{\beta}\log \rho\right),
		\end{align}
		where $\log \rho$ is the principal matrix logarithm of $\rho$. Since the matrix exponential is real-analytic, it follows by composition that $\beta \mapsto \rho^{(1-\beta)/\beta}$ is a real-analytic matrix-valued map on $(0,1)$. Multiplication by the fixed matrix $\sigma^{1/2}$ preserves real-analyticity, and therefore $\beta \mapsto X(\beta)$ is real-analytic as well.

		Next, for positive definite matrices the map
		\begin{align}
			(\beta,A) \longmapsto A^{\beta} = \exp\big(\beta \log A\big)
		\end{align}
		is real-analytic on $(0,1)$ times the cone of positive definite matrices, again because both the principal matrix logarithm and the matrix exponential are real-analytic there. Applying this to $A=X(\beta)$ shows that $\beta \mapsto X(\beta)^{\beta}$ is real-analytic. Since the trace is a linear map, it preserves real-analyticity, and hence so is
		\begin{align}
			\beta \longmapsto Q_{\Sand,\beta}(\sigma\|\rho) = \tr\big[X(\beta)^{\beta}\big].
		\end{align}
		Finally, $Q_{\Sand,\beta}(\sigma\|\rho) > 0$ for all $\beta \in (0,1)$, and the scalar logarithm is real-analytic on $(0,\infty)$. Therefore
		\begin{align}
			g(\beta)=\log Q_{\Sand,\beta}(\sigma\|\rho)
		\end{align}
		is real-analytic on $(0,1)$ as a composition of real-analytic maps. In particular $g \in C^1(0,1)$.

		\bigskip
		\textit{3.2) Verification for strict convexity.}  By~\cite[Lemma 3.1]{hayashi2017quantum}, $\beta \mapsto g(\beta)$ is convex on $(0,1)$. To upgrade convexity to strict convexity, it suffices to rule out that $g$ is affine, since a real-analytic convex function on an interval is either strictly convex or affine. We include the short argument for this fact for completeness.

		Suppose that $g$ were not strictly convex. There would exist $0<\beta_{0}<\beta_{1}<\beta_{2}<1$ such that
		\begin{align}
			g(\beta_{1})= \frac{\beta_{2}-\beta_{1}}{\beta_{2}-\beta_{0}}g(\beta_{0}) + \frac{\beta_{1}-\beta_{0}}{\beta_{2}-\beta_{0}}g(\beta_{2}).
		\end{align}
		Let $\ell$ be the affine function joining the two endpoint values,
		\begin{align}
			\ell(\beta):= \frac{\beta_{2}-\beta}{\beta_{2}-\beta_{0}}g(\beta_{0}) + \frac{\beta-\beta_{0}}{\beta_{2}-\beta_{0}}g(\beta_{2}).
		\end{align}
		By convexity, $g(\beta)\le \ell(\beta)$ for every $\beta \in [\beta_{0},\beta_{2}]$, and by assumption $g(\beta_{1})=\ell(\beta_{1})$. We claim that in fact $g(\beta)=\ell(\beta)$ for all $\beta \in [\beta_{0},\beta_{2}]$.

		Indeed, fix $\beta \in [\beta_{0},\beta_{1}]$. Then
		\begin{align}
			\beta_{1}= \lambda \beta + (1-\lambda)\beta_{2},
			\qquad
			\lambda := \frac{\beta_{2}-\beta_{1}}{\beta_{2}-\beta} \in (0,1].
		\end{align}
		Convexity gives
		\begin{align}
			g(\beta_{1}) \le \lambda g(\beta) + (1-\lambda)g(\beta_{2}).
		\end{align}
		On the other hand, since $g(\beta_{1})=\ell(\beta_{1})$, $\ell$ is affine, and $\ell(\beta_{2})=g(\beta_{2})$, we have
		\begin{align}
			\lambda \ell(\beta) + (1-\lambda)g(\beta_{2})
			&= \lambda \ell(\beta) + (1-\lambda)\ell(\beta_{2}) \\
			&= \ell\bigl(\lambda \beta + (1-\lambda)\beta_{2}\bigr) \\
			&= \ell(\beta_{1}) \\
			&= g(\beta_{1}) \\
			&\le \lambda g(\beta) + (1-\lambda)g(\beta_{2}).
		\end{align}
		Subtracting $(1-\lambda)g(\beta_{2})$ from both sides and dividing by $\lambda>0$, we obtain
		\begin{align}
			\ell(\beta) \le g(\beta).
		\end{align}
		Since already $g(\beta)\le \ell(\beta)$, we get $g(\beta)=\ell(\beta)$ on $[\beta_{0},\beta_{1}]$.  Hence $g$ is affine on this interval.

		Now $h:=g-\ell$ is real-analytic on $(0,1)$ and vanishes on the nonempty open interval $(\beta_{0},\beta_{1})$. By the identity theorem for real-analytic functions, $h$ must vanish identically on $(0,1)$; see, e.g.,~\cite[Corollary 1.2.5]{krantz_parks_analytic}. Therefore $g$ is affine on all of $(0,1)$. We have shown that a real-analytic convex function on $(0,1)$ is either strictly convex or affine.

		It remains to rule out the affine case. Assume therefore that $g$ is affine on $(0,1)$. Taking the limit $\beta \to 1^-$ and using $Q_{\Sand,1}(\sigma\|\rho)=\tr[\sigma]=1$, we obtain a constant $c \in \RR$ such that
		\begin{align}
			g(\beta)=c(\beta-1), \qquad \beta \in (0,1).
		\end{align}
		Equivalently,
		\begin{align}
			D_{\Sand,\beta}(\sigma\|\rho)=c, \qquad \beta \in (0,1).
		\end{align}
		Since $\sigma>0$, we have $\Pi_{\sigma}=I$, and 
		\begin{align}
			\lim_{\beta\to 0^+} D_{\Sand,\beta}(\sigma\|\rho) \leq \lim_{\beta\to 0^+} D_{\Petz,\beta}(\sigma\|\rho) = -\log \tr[\rho]=0,
		\end{align}
		so $c=0$. On the other hand, 
		\begin{align}
			\lim_{\beta\to 1^-} D_{\Sand,\beta}(\sigma\|\rho) = D(\sigma\|\rho).
		\end{align}
		Because $0< r < D_{\Rev}(\rho\|\sigma)$, we have $\rho \neq \sigma$, and therefore $D(\sigma\|\rho)>0$. This contradicts $c=0$. Hence $g$ is not affine, and therefore $g$ is strictly convex on $(0,1)$.

		Having verified all required assumptions of Theorem~\ref{thm: matching bounds} and identified the regularized divergence as $D_{\RSand,\alpha}^{\infty}(\sA\|\sB) = D_{\RSand,\alpha}(\rho\|\sigma)$, the claim~\eqref{eq: RSand operational meaning} follows.
	\end{proof}

	\begin{boxremark}[Necessity of $\sigma>0$ for strict convexity.]
		The assumption $\sigma > 0$ is essential for the strict convexity of $\beta \mapsto \log Q_{\Sand,\beta}(\sigma\|\rho)$. Indeed, let
		\begin{align}
			\rho = \begin{pmatrix} p & 0 \\ 0 & 1-p \end{pmatrix},
			\qquad
			\sigma = \begin{pmatrix} 1 & 0 \\ 0 & 0 \end{pmatrix} = |0\rangle\langle 0|,
			\qquad 0<p<1.
		\end{align}
		Then $\rho > 0$ but $\sigma$ is not strictly positive. For every $\beta \in (0,1)$,
		\begin{align}
			\rho^{\frac{1-\beta}{2\beta}}\sigma\rho^{\frac{1-\beta}{2\beta}}
			= p^{\frac{1-\beta}{\beta}} |0\rangle\langle 0|,
		\end{align}
		and therefore
		\begin{align}
			Q_{\Sand,\beta}(\sigma\|\rho)
			= \tr\!\left[\left(\rho^{\frac{1-\beta}{2\beta}}\sigma\rho^{\frac{1-\beta}{2\beta}}\right)^\beta\right]
			= p^{1-\beta}.
		\end{align}
		Hence
		\begin{align}
			\log Q_{\Sand,\beta}(\sigma\|\rho) = (1-\beta)\log p,
		\end{align}
		which is affine in $\beta$, not strictly convex. 
	\end{boxremark}

	The following corollary endows the reverse quantum relative entropy $D_{\Rev}(\rho\|\sigma)$ with the corresponding operational meaning in the Stein regime. Recall that the optimal Type-II error $\beta_{n,\ve}(\sA_n\|\sB_n)$ is defined in~\eqref{eq: optimal Type-II error sets 2}.

	\begin{boxcorollary}[Operational meaning of $D_{\Rev}$.]	
		\label{coro: RSand operational meaning Stein}
		Let $\rho, \sigma \in \density(\cH)$ with $\rho > 0$ and $\sigma > 0$, and let $\sA_n$, $\sB_n$ be as defined in~\eqref{eq: orbit alt set}.
		Then for any $\ve \in (0,1)$,
		\begin{align}\label{eq: RSand operational meaning Stein}
			\lim_{n\to\infty}-\frac{1}{n}\log \beta_{n,\ve}(\sA_n\|\sB_n)
			= D_{\Rev}(\rho\|\sigma).
		\end{align}
	\end{boxcorollary}
	\begin{proof}
		We prove the two directions separately.

		\medskip
		\textit{1) Achievability:} Fix any $0 < \delta < D_{\Rev}(\rho\|\sigma)$ and set $r := D_{\Rev}(\rho\|\sigma) - \delta$. So we have $0 < r < D_{\Rev}(\rho\|\sigma)$. 
		Since $ \sup_{\alpha \in (0,1)} D_{\RSand,\alpha}(\rho\|\sigma) = D_{\Rev}(\rho\|\sigma) > r$, there exists $\alpha' \in (0,1)$ such that $D_{\RSand,\alpha'}(\rho\|\sigma) > r$. Since $\frac{\alpha-1}{\alpha} < 0$ on $(0,1)$, the right hand side of~\eqref{eq: RSand operational meaning} is strictly positive. Theorem~\ref{thm: RSand operational meaning} therefore implies that there exists a sequence of tests $(M_n)_{n\in\mathbb{N}}$ such that
		\begin{align}
			\alpha(\sA_n, M_n) \le \ve, \qquad \beta(\sB_n, M_n) \le 2^{-nr},
		\end{align}
		for sufficiently large $n$.
		Hence $\beta_{n,\ve}(\sA_n\|\sB_n) \le 2^{-nr}$ for $n$ large, which yields
		\begin{align}
			\liminf_{n\to\infty}-\frac{1}{n}\log \beta_{n,\ve}(\sA_n\|\sB_n) \ge r = D_{\Rev}(\rho\|\sigma) - \delta.
		\end{align}
		Letting $\delta \to 0^+$ gives the achievability bound.

		\medskip
		\textit{2) Converse:} Since $\sB_n$ is closed under group actions $g^{\ox n}$, we have 
		\begin{align}
			\beta(\sB_n, M_n) = \sup_{\sigma_n \in \sB_n} \tr[M_n \sigma_n] \geq \sup_{\sigma_n \in \sB_n} \tr[M_n \cT^n_{\group}(\sigma_n)] = \beta(\cT^n_\group(\sB_n), M_n),
		\end{align}
		where the inequality follows as the average is no greater than the maximum. This implies that
		\begin{align}\label{eq: stein operational meaning tmp1}
			\beta_{n,\ve}(\sA_n\|\sB_n) \geq \beta_{n,\ve}(\sA_n\|\cT^n_\group(\sB_n)).
		\end{align}

		The converse part can then be shown by standard arguments applying to $\sA_n$ and $\cT^n_\group(\sB_n)$. 
		Let $\alpha > 1$. For any $\delta>0$, let $\rho_n \in \sA_n, \sigma_n\in\cT_{\group}^n(\sB_n)$ such that 
		\begin{align} 
			D_{\Sand,\alpha}(\rho_n\|\sigma_n) \leq D_{\Sand,\alpha}(\sA_n\|\cT_{\group}^n(\sB_n)) + \delta.
		\end{align} 
		By standard arguments, e.g.~\cite[Lemma 5]{cooney2016strong}, we have for any $0\leq M_n \leq I$, that
		\begin{align}
			\frac{1}{n} \log \left(1- \tr[(I-M_n)\rho_n]\right) \leq \frac{\alpha-1}{\alpha} \left(\frac{1}{n}D_{\Sand,\alpha}(\rho_n\|\sigma_n) + \frac{1}{n} \log \tr[M_n \sigma_n]\right).
		\end{align}
		Since $\tr[(I-M_n)\rho_n] \leq \alpha(\sA_n, M_n)$ and $\tr[M_n \sigma_n] \leq \beta(\cT_{\group}^n(\sB_n), M_n)$ by definitions, 
		\begin{align}
			\frac{1}{n} \log \left(1- \alpha(\sA_n, M_n)\right) \leq \frac{\alpha-1}{\alpha} \left(\frac{1}{n}D_{\Sand,\alpha}(\sA_n\|\cT_{\group}^n(\sB_n)) + \delta + \frac{1}{n} \log \beta(\cT_{\group}^n(\sB_n), M_n)\right).
		\end{align}
		Considering tests $0\leq M_n\leq I$ with $\alpha(\sA_n, M_n) \le \ve$, we get
		\begin{align}
			\frac{1}{n} \log (1-\ve) \le \frac{\alpha-1}{\alpha} \left(\frac{1}{n}D_{\Sand,\alpha}(\sA_n\|\cT_{\group}^n(\sB_n)) + \delta + \frac{1}{n} \log \beta_{n,\ve}(\sA_n\|\cT_{\group}^n(\sB_n))\right).
		\end{align}
		Rearranging the terms and utilizing $D_{\Petz,\alpha} \geq D_{\Sand,\alpha}$, we have 
		\begin{align}
			-\frac{1}{n} \log \beta_{n,\ve}(\sA_n\|\cT_{\group}^n(\sB_n)) \le \frac{\alpha}{\alpha-1} \frac{1}{n} \log \frac{1}{1-\ve} + \frac{1}{n}D_{\Petz,\alpha}(\sA_n\|\cT_{\group}^n(\sB_n)) + \delta.
		\end{align}
		Substituting~\eqref{eq: stein operational meaning tmp1} and taking the limit $n \to \infty$ and then $\delta \to 0$ gives
		\begin{align}
			\limsup_{n\to\infty}-\frac{1}{n} \log \beta_{n,\ve}(\sA_n\|\sB_n) \le D_{\Petz,\alpha}^{\infty}(\sA\|\cT_{\group}(\sB)).
		\end{align}
		As calculated in~\eqref{eq: operational meaning proof tmp1} and~\eqref{eq: operational meaning proof tmp2}, we have
		\begin{align}
			D_{\Petz,\alpha}^{\infty}(\sA\|\cT_{\group}(\sB)) = D_{\RSand,\alpha}^{\infty}(\sA\|\sB) = D_{\RSand,\alpha}(\rho\|\sigma),
		\end{align}
		Optimizing over $\alpha > 1$ gives
		\begin{align}
			\limsup_{n\to\infty}-\frac{1}{n} \log \beta_{n,\ve}(\sA_n\|\sB_n) \le \inf_{\alpha > 1} D_{\RSand,\alpha}(\rho\|\sigma) = D_{\Rev}(\rho\|\sigma).
		\end{align}
	This completes the converse part and hence the proof of Corollary~\ref{coro: RSand operational meaning Stein}.
	\end{proof}

	\begin{boxremark}[Restricted data-processing inequality.]
		The reverse sandwiched \Renyi divergence $D_{\RSand,\alpha}$ with $\alpha \in (1/2,1)$, and the reverse quantum relative entropy $D_{\Rev}$, are known to violate the standard data-processing inequality~\cite{audenaert2015alpha}. Corollary~\ref{coro: RSand operational meaning Stein} nevertheless yields a natural \emph{restricted} form of data processing under any channel that respects the group $\group$. Specifically, let $\Phi : \cL(\cH) \to \cL(\cH)$ be a $\group$-covariant channel, in the sense that for every $g \in \group$ there exists $\hat{g} \in \group$ with
		\begin{align}
			\Phi(g X g^\dagger) = \hat{g}\, \Phi(X)\, \hat{g}^\dagger \qquad \forall\, X \in \cL(\cH).
		\end{align}
		Then the data-processing inequality for $D_{\Rev}$ holds for the pair $(\rho,\sigma)$ under $\Phi$:
		\begin{align}\label{eq: restricted DPI}
			D_{\Rev}(\Phi(\rho)\|\Phi(\sigma)) \le D_{\Rev}(\rho\|\sigma).
		\end{align}
			Indeed, let $(\sA_n,\sB_n)$ be defined as in~\eqref{eq: orbit alt set} for $(\rho,\sigma,\group)$, and let $(\hat{\sA}_n,\hat{\sB}_n)$ be defined analogously for $(\Phi(\rho),\Phi(\sigma),\group)$. The covariance condition implies $\Phi^{\otimes n}(\sA_n) = \hat{\sA}_n$ and $\Phi^{\otimes n}(\sB_n) = \hat{\sB}_n$, so the data-processing inequality of the worst-case hypothesis-testing relative entropy yields
		\begin{align}
			\beta_{n,\ve}(\sA_n\|\sB_n) \le \beta_{n,\ve}(\Phi^{\otimes n}(\sA_n)\|\Phi^{\otimes n}(\sB_n)) = \beta_{n,\ve}(\hat{\sA}_n\|\hat{\sB}_n).
		\end{align}
		Taking $-\frac{1}{n}\log$, sending $n\to\infty$, and applying Corollary~\ref{coro: RSand operational meaning Stein} to both sides establishes~\eqref{eq: restricted DPI}.
	\end{boxremark}

	\subsubsection*{Free time evolution under a Hamiltonian}

	The full diagonal-unitary group in~\eqref{eq: group def corollary} represents the most general dephasing noise in the energy eigenbasis. A physically important specialization is the one-parameter family generated by the Hamiltonian itself. We now show that, under a rational independence condition, this restriction does not change the asymptotic exponent. Consider a quantum system governed by a Hamiltonian $H$ with energy eigenbasis $\{|E_j\rangle\}$ and eigenvalues $\{E_j\}$. Its thermal equilibrium state at inverse temperature $\beta$ is
	\begin{align}
		\rho = \frac{e^{-\beta H}}{\tr[e^{-\beta H}]} = \sum_{j=1}^{m}\lambda_j|E_j\rangle\langle E_j|,
		\qquad \lambda_j = \frac{e^{-\beta E_j}}{\tr[e^{-\beta H}]}.
	\end{align}
	To connect this with physical time evolution, note that
	\begin{align}
		\log\lambda_j = -\beta E_j - \log Z_\beta, \qquad Z_\beta := \tr[e^{-\beta H}],
	\end{align}
	so for any $t \in \RR$,
	\begin{align}
		e^{it\log\lambda_j} = e^{-it\log Z_\beta}\,e^{-i(\beta t)E_j}.
	\end{align}
	Hence
	\begin{align}
		\sum_j e^{it\log\lambda_j}|E_j\rangle\langle E_j|
		= e^{-it\log Z_\beta}\,e^{-iH(\beta t)}.
	\end{align}
	The prefactor $e^{-it\log Z_\beta}$ is a global phase and therefore cancels under conjugation, while $\beta t$ is only a reparametrization of time. Thus conjugation by the physical evolution $e^{-iHt}$ is equivalent to conjugation by the diagonal one-parameter family
	\begin{align}\label{eq: time evolution group}
		\group' := \left\{g_t := \sum_{j=1}^{m} e^{it\log\lambda_j}|E_j\rangle\langle E_j|  \,:\, t \geq 0 \right\}.
	\end{align}
	We then define the null and alternative hypotheses by
	\begin{align}
		\sA_n' &:= \{\rho^{\otimes n}\}, \quad \text{and}\quad
		\sB_n' := \left\{g_t^{\otimes n}\sigma^{\otimes n}(g_t^\dagger)^{\otimes n} : g_t \in \group'\right\}. \label{eq: time orbit alt set}
	\end{align}

	Assume that the eigenvalues of $\rho$ are non-degenerate and that $\{\log\lambda_j\}_{j=1}^m$ is rationally independent. Then $\group'$ is dense in the full diagonal-unitary group $\group$, so each $\group'$-orbit is dense in the corresponding $\group$-orbit. As shown below, this density already implies equality of the optimal worst-case Type-I errors, and hence of the asymptotic error exponents. Although stated as an assumption, rational independence is a generic condition: the set of vectors $\RR^m$ that fail to be rationally independent is a countable union of hyperplanes, and therefore has Lebesgue measure zero in $\RR^m$.

	\begin{boxcorollary}[Alternative operational meanings of $D_{\RSand,\alpha}$ and $D_{\Rev}$.]
		\label{coro: RSand time evolution}
		Let $\rho, \sigma \in \density(\cH)$ with $\rho > 0$, $\sigma > 0$.
		Assume further that $\{\log\lambda_j\}_{j=1}^{m}$ are rationally independent, i.e., $\sum_{j=1}^{m} n_j\log\lambda_j = 0$ with $n_j \in \ZZ$ implies $n_j = 0$ for all $j$.
		Let $\sA_n'$ and $\sB_n'$ be as defined in~\eqref{eq: time orbit alt set}.
		Then for any $0 < r < D_{\Rev}(\rho\|\sigma)$, it holds that
		\begin{align}\label{eq: RSand time evolution}
			\lim_{n\to\infty}-\frac{1}{n}\log \alpha_{n,r}(\sA_n'\|\sB_n')
			= \sup_{\alpha \in (0,1)} \frac{\alpha-1}{\alpha}\big(r - D_{\RSand,\alpha}(\rho\|\sigma)\big).
		\end{align}
		For any $\ve \in (0,1)$, it holds that
		\begin{align}\label{eq: RSand time evolution Stein}
			\lim_{n\to\infty}-\frac{1}{n}\log \beta_{n,\ve}(\sA_n'\|\sB_n')
			= D_{\Rev}(\rho\|\sigma).
		\end{align}
	\end{boxcorollary}

	\begin{proof}
		We compare the testing problem defined by the time-evolution family $(\sA_n',\sB_n')$ with the one defined by the full diagonal-unitary group $(\sA_n,\sB_n)$ in Theorem~\ref{thm: RSand operational meaning}. Since $\sA_n' = \sA_n = \{\rho^{\otimes n}\}$, it suffices to show that
		\begin{align}
			\alpha_{n,r}(\sA_n'\|\sB_n') = \alpha_{n,r}(\sA_n\|\sB_n)
		\end{align}
		for every $n \in \NN$ and every $r > 0$.

		The key point is that the one-parameter family $\group'$ is dense in the full diagonal-unitary group $\group$. To see this, identify $\group$ with the torus $\mathbb{T}^{m} := \{(z_1,\dots,z_m) : |z_j|=1\}$ via
		\begin{align}
			(z_1,\dots,z_m) \longleftrightarrow \sum_{j=1}^{m} z_j |E_j\rangle\langle E_j|.
		\end{align}
		Under this identification, the family $\group'$ is exactly the image of the continuous homomorphism
		\begin{align}
			\Phi\colon \RR &\to \mathbb{T}^{m}, \\
			t &\mapsto \bigl(e^{it\log\lambda_1},\dots,e^{it\log\lambda_{m}}\bigr).
		\end{align}
		Kronecker's density theorem on the torus (e.g.~\cite[Theorem 443]{hardy1965introduction} or~\cite{bailleul2022explicit}) states that if $a_1,\dots,a_m$ are rationally independent, then the set
		\begin{align}
			\left\{(e^{ita_1},\dots,e^{ita_m}) : t \geq 0\right\}
		\end{align}
		is dense in $\mathbb{T}^m$. Applying this with $a_j = \log\lambda_j$, we conclude that for every target phase vector $(e^{i\theta_1},\dots,e^{i\theta_m}) \in \mathbb{T}^{m}$ and every $\varepsilon > 0$, there exists some $t \geq 0$ such that each coordinate $e^{it\log\lambda_j}$ is within $\varepsilon$ of $e^{i\theta_j}$. Equivalently, every element of $\group$ can be approximated arbitrarily well by some element $g_t \in \group'$. Hence $\group'$ is dense in $\group$.

		We now pass this density from the unitaries to the alternative states. Fix $n \in \NN$ and $g \in \group$. Choose a sequence $t_k$ such that $g_{t_k} \to g$. Since matrix multiplication is continuous in finite dimensions, we obtain
		\begin{align}
			g_{t_k}^{\otimes n}\sigma^{\otimes n}(g_{t_k}^\dagger)^{\otimes n}
			\longrightarrow
			g^{\otimes n}\sigma^{\otimes n}(g^\dagger)^{\otimes n}.
		\end{align}
		Therefore $\sB_n'$ is dense in $\sB_n$ for every $n$.

		Now fix any test $0 \le M_n \le I$. The map
		\begin{align}
			\omega \mapsto \tr[\omega M_n]
		\end{align}
		is continuous on the state space, so taking the supremum over the dense subset $\sB_n'$ or over its closure $\sB_n$ gives the same worst-case Type-II error:
		\begin{align}
			\sup_{\omega \in \sB_n'} \tr[\omega M_n]
			= \sup_{\omega \in \sB_n} \tr[\omega M_n].
		\end{align}
		Hence the feasible tests in the optimization problems defining $\alpha_{n,r}(\sA_n'\|\sB_n')$ and $\alpha_{n,r}(\sA_n\|\sB_n)$ coincide, and therefore these two quantities are equal for every $n$ and $r$.
		The two testing problems thus have the same asymptotic error exponent, and the claim follows from Theorem~\ref{thm: RSand operational meaning}. The same argument applies to the Stein regime as well, Corollary~\ref{coro: RSand operational meaning Stein} then yields~\eqref{eq: RSand time evolution Stein}.
	\end{proof}

	\section{Discussion}
	\label{sec: discussion}

	We established quantum Hoeffding bounds for composite hypothesis testing under a new set of structural assumptions that replaces tensor-product stability with permutation symmetry of the state sequences, together with regularity properties of the associated regularized \Renyi divergences. As a notable application, these bounds yield the direct operational interpretation of the reverse sandwiched \Renyi divergence and the reverse quantum relative entropy in composite quantum hypothesis testing. Along the way, we developed several technical tools that may be of independent interest, including the construction of a universal test via pinching in Lemma~\ref{lem: universal test pinching}, the construction of a universally tight state pair in Lemma~\ref{lem: averaged state pair}, and the equivalence between twirling and pinching established in Lemma~\ref{lem: twirling equals pinching}.

	Several directions remain open. First, our framework still relies on the assumption of convexity and therefore does not yet cover the fully general composite i.i.d.\ setting; removing or weakening this requirement, so as to pin down the exact error exponent in this regime, is an interesting open problem. Second, as exemplified in this work, the composite setting offers a natural arena for uncovering operational interpretations of quantum divergences; it would therefore be valuable to explore the possibility of such interpretations for divergences whose operational meaning is currently unknown, such as the Belavkin--Staszewski relative entropy~\cite{belavkin1982c,fang2021geometric}. Third, since operational interpretations are by no means unique, it would also be interesting to identify further operational roles of the reverse sandwiched \Renyi divergence in other contexts. Finally, the error-exponent bounds developed here may find applications in broader information-theoretic tasks, which we leave for future investigation.

	\paragraph{Acknowledgements.} 
	M.H. is supported in part by the General R\&D Projects of 1+1+1 CUHK-CUHK(SZ)-GDST Joint Collaboration Fund (Grant No. GRDP2025-022), the Guangdong Provincial Quantum Science Strategic Initiative (Grant No. GDZX2505003),
	the Shenzhen International Quantum Academy (Grant No. SIQA2025KFKT07). 
	K.F. is supported in part by the National Natural Science Foundation of China (Grant No. 12404569 and 92470113), the Shenzhen Science and Technology Program (Grant No. QNXMB20250701091826036 and JCYJ20240813113519025), the Shenzhen International Quantum Academy (Grant No. SIQA2025KFKT03), the Shenzhen Fundamental Research Program (Grant No. JCYJ20241202124023031), the General R\&D Projects of 1+1+1 CUHK-CUHK(SZ)-GDST Joint Collaboration Fund (Grant No. GRDP2025-022), and the University Development Fund (Grant No. UDF01003565). 

	\paragraph{Data availability statement.} No datasets were generated or analysed during the current study. There is no conflict of interest in this work.

	\bibliographystyle{alpha_abbrv}
	\bibliography{Bib}

\newcommand{\etalchar}[1]{$^{#1}$}
\begin{thebibliography}{ACMT{\etalchar{+}}07}

\bibitem[ACMT{\etalchar{+}}07]{audenaert2007discriminating}
K.~M. Audenaert, J.~Calsamiglia, R.~Munoz-Tapia, E.~Bagan, L.~Masanes, A.~Acin,
  and F.~Verstraete.
\newblock {Discriminating states: The quantum Chernoff bound}.
\newblock {\em Physical Review Letters}, 98(16):160501, 2007.

\bibitem[AD15]{audenaert2015alpha}
K.~M. Audenaert and N.~Datta.
\newblock $\alpha$-z-{R}{\'e}nyi relative entropies.
\newblock {\em Journal of Mathematical Physics}, 56(2), 2015.

\bibitem[ANSV08]{audenaert2008asymptotic}
K.~M. Audenaert, M.~Nussbaum, A.~Szko{\l}a, and F.~Verstraete.
\newblock Asymptotic error rates in quantum hypothesis testing.
\newblock {\em Communications in Mathematical Physics}, 279(1):251--283, 2008.

\bibitem[Bai22]{bailleul2022explicit}
A.~Bailleul.
\newblock Explicit {Kronecker--Weyl} theorems and applications to prime number
  races.
\newblock {\em Research in Number Theory}, 8(3):43, 2022.

\bibitem[BBH21]{berta2021composite}
M.~Berta, F.~G. Brandão, and C.~Hirche.
\newblock On composite quantum hypothesis testing.
\newblock {\em Communications in Mathematical Physics}, 385(1):55--77, 2021.

\bibitem[BHO{\etalchar{+}}13]{brandao2013resource}
F.~G. Brandao, M.~Horodecki, J.~Oppenheim, J.~M. Renes, and R.~W. Spekkens.
\newblock Resource theory of quantum states out of thermal equilibrium.
\newblock {\em Physical Review Letters}, 111(25):250404, 2013.

\bibitem[BP10]{brandao2010generalization}
F.~G. Brandao and M.~B. Plenio.
\newblock A generalization of quantum {Stein}'s lemma.
\newblock {\em Communications in Mathematical Physics}, 295(3):791--828, 2010.

\bibitem[BS82]{belavkin1982c}
V.~P. Belavkin and P.~Staszewski.
\newblock C*-algebraic generalization of relative entropy and entropy.
\newblock In {\em Annales de l'institut Henri Poincar{\'e}. Section A, Physique
  Th{\'e}orique}, volume~37, pages 51--58, 1982.

\bibitem[BV04]{boyd2004convex}
S.~Boyd and L.~Vandenberghe.
\newblock {\em Convex optimization}.
\newblock Cambridge University Press, 2004.

\bibitem[CMW16]{cooney2016strong}
T.~Cooney, M.~Mosonyi, and M.~M. Wilde.
\newblock Strong converse exponents for a quantum channel discrimination
  problem and quantum-feedback-assisted communication.
\newblock {\em Communications in Mathematical Physics}, 344(3):797--829, 2016.

\bibitem[DZ10]{Dembo2010}
A.~Dembo and O.~Zeitouni.
\newblock {\em Large Deviations Techniques and Applications}.
\newblock Springer Berlin Heidelberg, 2010.

\bibitem[FF21]{fang2021geometric}
K.~Fang and H.~Fawzi.
\newblock {Geometric R{\'e}nyi divergence and its applications in quantum
  channel capacities}.
\newblock {\em Communications in Mathematical Physics}, 384(3):1615--1677,
  2021.

\bibitem[FFF24]{fang2024generalized}
K.~Fang, H.~Fawzi, and O.~Fawzi.
\newblock Generalized quantum asymptotic equipartition.
\newblock {\em arXiv preprint arXiv:2411.04035}, 2024.

\bibitem[FGP{\etalchar{+}}22]{fawzi2022lifting}
H.~Fawzi, J.~Gouveia, P.~A. Parrilo, J.~Saunderson, and R.~R. Thomas.
\newblock Lifting for simplicity: Concise descriptions of convex sets.
\newblock {\em SIAM Review}, 64(4):866--918, 2022.

\bibitem[FGW25]{fang2025towards}
K.~Fang, G.~Gour, and X.~Wang.
\newblock Towards the ultimate limits of quantum channel discrimination and
  quantum communication.
\newblock {\em Science China Information Sciences}, 68(8):180509, 2025.

\bibitem[FH26]{fang2025error}
K.~Fang and M.~Hayashi.
\newblock Error exponents of quantum state discrimination with composite
  correlated hypotheses.
\newblock {\em IEEE Transactions on Information Theory}, 2026.
\newblock 10.1109/TIT.2026.3684314.

\bibitem[Hay07]{hayashi2007error}
M.~Hayashi.
\newblock Error exponent in asymmetric quantum hypothesis testing and its
  application to classical-quantum channel coding.
\newblock {\em Physical Review A—Atomic, Molecular, and Optical Physics},
  76(6):062301, 2007.

\bibitem[Hay17]{hayashi2017quantum}
M.~Hayashi.
\newblock {\em Quantum Information Theory}.
\newblock Graduate Texts in Physics. Springer, Berlin, Heidelberg, 2 edition,
  2017.

\bibitem[Hay25a]{Hayashi2025another}
M.~Hayashi.
\newblock Another quantum version of {Sanov} theorem.
\newblock In {\em Annales Henri Poincar{\'e}}, pages 1--22. Springer, 2025.

\bibitem[Hay25b]{hayashi2025general}
M.~Hayashi.
\newblock General detectability measure.
\newblock {\em Communications in Mathematical Physics}, 406(12):289, 2025.

\bibitem[HI25]{Hayashi2025entanglement}
M.~Hayashi and Y.~Ito.
\newblock Entanglement measures for detectability.
\newblock {\em IEEE Transactions on Information Theory}, 71(6):4385–4405,
  April 2025.

\bibitem[HMO07]{hiai2007large}
F.~Hiai, M.~Mosonyi, and T.~Ogawa.
\newblock Large deviations and {C}hernoff bound for certain correlated states
  on a spin chain.
\newblock {\em Journal of Mathematical Physics}, 48(12), 2007.

\bibitem[HMO08]{hiai2008error}
F.~Hiai, M.~Mosonyi, and T.~Ogawa.
\newblock Error exponents in hypothesis testing for correlated states on a spin
  chain.
\newblock {\em Journal of Mathematical Physics}, 49(3), 2008.

\bibitem[HO13]{horodecki2013fundamental}
M.~Horodecki and J.~Oppenheim.
\newblock Fundamental limitations for quantum and nanoscale thermodynamics.
\newblock {\em Nature Communications}, 4(1):2059, 2013.

\bibitem[HP91]{hiai1991proper}
F.~Hiai and D.~Petz.
\newblock The proper formula for relative entropy and its asymptotics in
  quantum probability.
\newblock {\em Communications in Mathematical Physics}, 143:99--114, 1991.

\bibitem[HT16]{hayashi2016correlation}
M.~Hayashi and M.~Tomamichel.
\newblock {Correlation detection and an operational interpretation of the
  {R{\'e}nyi} mutual information}.
\newblock {\em Journal of Mathematical Physics}, 57(10), 2016.

\bibitem[HW65]{hardy1965introduction}
G.~Hardy and E.~Wright.
\newblock {\em An Introduction to the Theory of Numbers. 4th Ed}.
\newblock Oxford science publications. Clarendon Press, 1965.

\bibitem[KL51]{kullback1951information}
S.~Kullback and R.~A. Leibler.
\newblock On information and sufficiency.
\newblock {\em The Annals of Mathematical Statistics}, 22(1):79--86, 1951.

\bibitem[KP92]{krantz_parks_analytic}
S.~G. Krantz and H.~R. Parks.
\newblock {\em A Primer of Real Analytic Functions}.
\newblock Birkh\"auser, 1992.

\bibitem[KW24]{khatri2024principlesquantumcommunicationtheory}
S.~Khatri and M.~M. Wilde.
\newblock Principles of quantum communication theory: A modern approach.
\newblock {\em arXiv preprint arXiv:2011.04672}, 2024.

\bibitem[Lam25a]{lami2024solutiongeneralisedquantumsteins}
L.~Lami.
\newblock {A solution of the generalized quantum Stein’s lemma}.
\newblock {\em IEEE Transactions on Information Theory}, 71(6):4454–4484,
  June 2025.

\bibitem[Lam25b]{lami2025doubly}
L.~Lami.
\newblock A doubly composite {Chernoff-Stein} lemma and its applications.
\newblock {\em arXiv preprint arXiv:2510.06342}, 2025.

\bibitem[LBCR{\etalchar{+}}24]{lipka2024quantum}
P.~Lipka-Bartosik, C.~T. Chubb, J.~M. Renes, M.~Tomamichel, and K.~Korzekwa.
\newblock Quantum dichotomies and coherent thermodynamics beyond first-order
  asymptotics.
\newblock {\em PRX Quantum}, 5:020335, May 2024.

\bibitem[MLDS{\etalchar{+}}13]{muller2013quantum}
M.~M{\"u}ller-Lennert, F.~Dupuis, O.~Szehr, S.~Fehr, and M.~Tomamichel.
\newblock On quantum {R}{\'e}nyi entropies: {A} new generalization and some
  properties.
\newblock {\em Journal of Mathematical Physics}, 54(12), 2013.

\bibitem[MO15a]{mosonyi2015quantum}
M.~Mosonyi and T.~Ogawa.
\newblock {Quantum hypothesis testing and the operational interpretation of the
  quantum R{\'e}nyi relative entropies}.
\newblock {\em Communications in Mathematical Physics}, 334:1617--1648, 2015.

\bibitem[MO15b]{mosonyi2015two}
M.~Mosonyi and T.~Ogawa.
\newblock Two approaches to obtain the strong converse exponent of quantum
  hypothesis testing for general sequences of quantum states.
\newblock {\em IEEE Transactions on Information Theory}, 61(12):6975--6994,
  2015.

\bibitem[MSW22]{Mosonyi_2022}
M.~Mosonyi, Z.~Szilagyi, and M.~Weiner.
\newblock On the error exponents of binary state discrimination with composite
  hypotheses.
\newblock {\em IEEE Transactions on Information Theory}, 68(2):1032–1067,
  February 2022.

\bibitem[Nag06]{nagaoka2006converse}
H.~Nagaoka.
\newblock The converse part of the theorem for quantum hoeffding bound.
\newblock {\em arXiv preprint quant-ph/0611289}, 2006.

\bibitem[NS09]{Nussbaum2009}
M.~Nussbaum and A.~Szko{\l}a.
\newblock {The Chernoff lower bound for symmetric quantum hypothesis testing}.
\newblock {\em The Annals of Statistics}, 37(2):1040--1057, 2009.

\bibitem[ON00]{Ogawa2000}
T.~Ogawa and H.~Nagaoka.
\newblock {Strong converse and {Stein's} lemma in quantum hypothesis testing}.
\newblock {\em IEEE Transactions on Information Theory}, 46(7):2428--2433, feb
  2000.

\bibitem[Pet86]{petz1986quasi}
D.~Petz.
\newblock Quasi-entropies for finite quantum systems.
\newblock {\em Reports on Mathematical Physics}, 23(1):57--65, 1986.

\bibitem[Roc70]{rockafellar}
R.~T. Rockafellar.
\newblock {\em Convex Analysis}.
\newblock Princeton University Press, 1970.

\bibitem[Sha48]{shannon1948mathematical}
C.~E. Shannon.
\newblock A mathematical theory of communication.
\newblock {\em The Bell System Technical Journal}, 27(3):379--423, 1948.

\bibitem[Ume62]{umegaki1962conditional}
H.~Umegaki.
\newblock Conditional expectation in an operator algebra, {IV} (entropy and
  information).
\newblock In {\em Kodai Mathematical Seminar Reports}, volume~14, pages 59--85.
  Department of Mathematics, Tokyo Institute of Technology, 1962.

\bibitem[Wat18]{watrous2018theory}
J.~Watrous.
\newblock {\em The theory of quantum information}.
\newblock Cambridge University Press, 2018.

\bibitem[WDH26]{warsi2025generalization}
N.~A. Warsi, A.~Dasgupta, and M.~Hayashi.
\newblock Generalization bounds for quantum learning via {R}\'{e}nyi
  divergences.
\newblock {\em IEEE Transactions on Information Theory}, 2026.
\newblock 10.1109/TIT.2026.3683678.

\bibitem[Wil11]{wilde2011classical}
M.~M. Wilde.
\newblock From classical to quantum {S}hannon theory.
\newblock {\em arXiv preprint arXiv:1106.1445}, 2011.

\bibitem[WT24]{watanabe2024black}
K.~Watanabe and R.~Takagi.
\newblock Black box work extraction and composite hypothesis testing.
\newblock {\em Physical Review Letters}, 133(25):250401, 2024.

\bibitem[WWY14]{wilde2014strong}
M.~M. Wilde, A.~Winter, and D.~Yang.
\newblock Strong converse for the classical capacity of entanglement-breaking
  and {H}adamard channels via a sandwiched {R}{\'e}nyi relative entropy.
\newblock {\em Communications in Mathematical Physics}, 331:593--622, 2014.

\bibitem[ZF26]{zhang2026quantum}
M.~Zhang and K.~Fang.
\newblock Quantum thermodynamics with uncertain equilibrium.
\newblock {\em arXiv preprint arXiv:2604.13524}, 2026.

\end{thebibliography}

\end{document}